\documentclass[10pt,journal,compsoc]{IEEEtran}

\ifCLASSOPTIONcompsoc
  \usepackage[nocompress]{cite}
\else
  \usepackage{cite}
\fi

\usepackage[dvips]{graphicx}
\usepackage[caption=false]{subfig}

\usepackage{amsmath}
\usepackage{amssymb}
\usepackage{amsthm}
\usepackage[binary-units = true]{siunitx}

\usepackage[T1]{fontenc}
\usepackage{xurl}
\usepackage{color,soul}
\usepackage{xcolor}
\usepackage{ragged2e}
\usepackage{changes}
\usepackage{multirow}

\newtheorem{theorem}{Theorem}

\makeatother

\usepackage{algorithm}
\usepackage[noend]{algorithmic}	
\begin{document}

\title{Co-evolution of Viral Processes and Structural Stability in Signed Social Networks}

\author{Temirlan~Kalimzhanov,
        Amir~Haji~Ali~Khamseh'i,
        Aresh~Dadlani,~\IEEEmembership{Senior Member,~IEEE,}
        Muthukrishnan~Senthil~Kumar,~\IEEEmembership{Member,~IEEE,}
        and~Ahmad~Khonsari,~\IEEEmembership{Member,~IEEE}\vspace{-0.5em}% <-this % stops a space
\IEEEcompsocitemizethanks{
			\IEEEcompsocthanksitem T. Kalimzhanov and A. Dadlani are with the Department of Electrical and Computer Engineering, Nazarbayev University (NU), Nursultan 010000, Kazakhstan. 
			E-mails: \{temirlan.kalimzhanov, aresh.dadlani\}@nu.edu.kz
			\IEEEcompsocthanksitem A. H. A. Khamseh'i and A. Khonsari are with the School of Electrical and Computer Engineering, College of Engineering, University of Tehran, Tehran 1439957131, Iran. 
			E-mails: \{khamse, a\_khonsari\}@ut.ac.ir
\IEEEcompsocthanksitem M. S. Kumar is with the Department of Applied Mathematics and Computational Sciences, PSG College of Technology, Coimbatore 641004, India.
			E-mail: msk.amcs@psgtech.ac.in}% <-this % stops an unwanted space
}

\markboth{IEEE TRANSACTIONS ON KNOWLEDGE AND DATA ENGINEERING,~Vol.~XX, No.~XX, 2021}%
{Shell \MakeLowercase{\textit{et al.}}: Bare Demo of IEEEtran.cls for Computer Society Journals}

\IEEEtitleabstractindextext{\justify%
\begin{abstract}
\fontdimen2\font=0.5ex
Prediction and control of spreading processes in social networks (SNs) are closely tied to the underlying connectivity patterns. Contrary to most existing efforts that exclusively focus on positive social user interactions, the impact of contagion processes on the temporal evolution of signed SNs (SSNs) with distinctive friendly (positive) and hostile (negative) relationships yet, remains largely unexplored. In this paper, we study the interplay between social link polarity and propagation of viral phenomena coupled with user alertness. In particular, we propose a novel energy model built on Heider’s balance theory that relates the stochastic susceptible-alert-infected-susceptible epidemic dynamical model with the structural balance of SSNs to substantiate the trade-off between social tension and epidemic spread. Moreover, the role of hostile social links in the formation of disjoint friendly clusters of alerted and infected users is analyzed. Using three real-world SSN datasets, we further present a time-efficient algorithm to expedite the energy computation in our Monte-Carlo simulation method and show compelling insights on the effectiveness and rationality of user awareness and initial network settings in reaching structurally balanced local and global network energy states.
\end{abstract}

\begin{IEEEkeywords}
Signed networks, epidemic process, balance theory, awareness, continuous-time Markov chain, energy function.\vspace{-0.2em}
\end{IEEEkeywords}}

		\maketitle

\IEEEdisplaynontitleabstractindextext
\IEEEpeerreviewmaketitle

\IEEEraisesectionheading{\section{Introduction}
\label{sec:introduction}}
\fontdimen2\font=0.39ex
\vspace{-0.2em}
\IEEEPARstart{Q}{uantitative} analysis of epidemic processes~such~as infectious diseases, malware codes, and rumors spreading over physical and online social networks (SNs)~has stimulated intense research activities \cite{Britton_2020,Huang_2020}. Owing to the pervasive use of social media and the abundance of data~extracted from several such networks, which for long were~merely~unavailable, the theoretical perception of epidemic dynamics driven by nodal interactions has refined substantially in recent years \cite{Dabarov_2020, Shafiei_2020}. While the vast majority of research has scrutinized only positive social relationships, user pairs may also signify enmity or distrust as perceived in reality. Subsequently, a~user may decisively decline to interact with a hostile contact and avoid involvement in further spread of the viral process \cite{Leskovec_2010}. Accounting for heterogeneous social interactions is thus,~crucial in characterizing social link valence evolution under the influence of individual user's attitudes towards viral spread.

Unlike conventional networks, signed SNs (SSNs) evolve based on the structural balance theory, pioneered by Heider \cite{Heider_1958}, where the relationship between any two users in a triad ($3$-clique) can be impacted by the third user \cite{Sun_2020}. That is to say, the theory posits that if ``\textit{the friend of my friend is my friend}'' and ``\textit{the friend of my enemy is my enemy}'', then the resulting triad will be balanced and will constitute an odd number of friendly links. Evidently, SSNs converge to structurally balanced states with minimum social tension by flipping the link polarity to maximize the number of balanced triads \cite{Hedayatifar_2017}.

In the jargon of networked epidemics, a handful of works focus on edge sign reconfiguration under the effect of evolving user states. The conditions to attain opinion convergence in generic SNs are obtained in \cite{Altafini_2013} using monotone dynamical systems. Further extended in \cite{Shi2015}, Shi \emph{et al.} analyze the asymptotic user state evolution affected~by deterministic weights on pairwise interactions by~formulating a relative-state-flipping model for consensus dynamics in random SSNs and prove the conditions leading to almost sure convergence and divergence. Their analysis assumes that the initial network structure is always balanced~which in truth, may not always be the case. Saeedian \emph{et al.} \cite{Saeedian_2017}~study the non-trivial coupled dynamics over a complete signed graph using an energy function. The authors adopt the susceptible-infected (SI) epidemic model to study the local and global energy minima of the system irrespective of the possibility of recovery to susceptibility~or epidemic~alertness. Lee \emph{et al.}\,\cite{Lee_2019} then~introduce an adaptive susceptible-infected-susceptible (SIS) model to reinforce transitivity by rewiring the links between susceptible and infected nodes rather than their signs. Though insightful, the emergent behaviors of the parallel processes in \cite{Lee_2019} are limited to the population level and do not address the microscopic dynamics inherent in user interactions. Zhang \emph{et\,al.} \cite{Zhao2017} present an approximation algorithm for the minimum partial positive influence seeding problem in viral marketing. Moreover, Li \emph{et~al.} \cite{Li_2021} propose a~non-stochastic computational~model~for~maximizing~polarity-related~linear influence diffusion in SNs. However, neither\,\cite{Zhao2017} nor \cite{Li_2021} subsume Heider's theory in the network structural evolution.

Thus far, there exists no work that investigates the~intriguing co-evolution of generic SN structures and~epidemic dynamics of a reversible process in conjunction with user alertness. Absent in classical epidemic models,~\textit{awareness} towards viral processes is an intrinsic human response~that plays the role of~a~natural~immunization strategy. Such~change in human behavior can be induced by learning about the~contagion spread from others without having to encounter it~firsthand. This hence, results in a coupled situation where an~infectious person and information about its presence spread~simultaneously when humans react to the presence of the infection. In fact, raising awareness is a widely-practiced~control strategy in dynamical systems as it alters the progression~of the viral spread \cite{Faryad_2011,Sahneh_2019}. Besides the analytical merits, such a~refined projection model may serve beneficial to network~administrators, social~influencers, and decision makers in devising optimal resource allocation policies. In view of this research gap, the main~contributions of this work are as follows:\vspace{-0.4em}
\begin{itemize}
	\item Inspired by Heider’s balance\,theory, a novel energy-based framework is proposed to jointly~minimize~the number of unbalanced triads that contribute~to~the~social tension while mitigating~the~viral spread~in SNs. To capture users’ response to such~processes, we formulate the susceptible-alert-infected-susceptible (SAIS) epidemic model \cite{Faryad_2011} as a continuous-time Markov chain (CTMC) and derive the stationary probabilities to investigate the virtues of promoting awareness on the network structural evolution.
	\item By incorporating a tuning parameter, we then analyze cases for which the initial fraction of positive links and the initially infected users induce \textit{natural immunization} by segregating the alerted and infected users into two clusters interconnected via unfriendly links.
	\item Our model is evaluated on three real SSN datasets~by employing a time-efficient Monte Carlo simulation method under different parametric~settings.\vspace{-1.1em}
\end{itemize}
%The rest of this paper is organized as follows: Section~\ref{sec2} outlines the proposed network energy framework, followed by the definition of the corresponding CTMC and stationary probabilities in Section~\ref{sec3}. In Section~\ref{sec4}, an efficient algorithm for faster energy computation in the employed Monte Carlo method is presented. Experimental validation of the model is discussed in Section~\ref{sec5}. Finally, Section~\ref{sec6} concludes the paper.\vspace{-0.4em}% with directions for potential future works.
\begin{figure}[!t]
	\centering
	\includegraphics[width=2.5in]{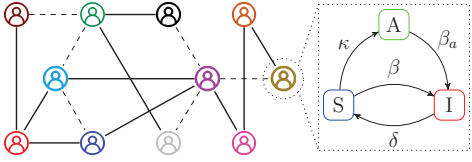}
	\vspace{-0.8em}
	\caption{Schematic of the SAIS spreading model over an SN with positive (solid) and negative (dashed) social relationships.}
	\label{fig1}
	\vspace{-0.4em}
\end{figure}

\section{The Coupled Network Model}
\label{sec2}
\fontdimen2\font=0.45ex
In this section, we first formulate the stochastic SAIS model re-purposed for general spreading processes in our analysis. The proposed energy model is then detailed subsequently.\vspace{-0.6em}

\subsection{Stochastic Epidemic Model Description}
\label{sec2.1}
\fontdimen2\font=0.43ex
We consider an undirected SSN, represented by the graph $\mathcal{G}_t \!=\! (\mathcal{V},\mathcal{E}_t)$, with a set $\mathcal{V} \!=\! \{1,2,\ldots,n\}$ of $n$ users that~form friendly ($1$), hostile ($-1$), or no ($0$) social links. The link~polarity of user pair ${(i,j) \!\in\! \mathcal{E}_t}$ at any given time ${t}$ is denoted~by ${\mathcal{A}_{i,j}(t) \!\in\! \{-1,0,1\}}$ \cite{Saeedian_2017}. For epidemic spreading over $\mathcal{G}_t$, the SAIS model in \figurename{~\ref{fig1}} is used, where each user is in the susceptible (S), alert (A), or infected (I) state at time~$t$. User $i$ is said to be susceptible if he/she is completely unaware of the spreading process. Since these processes do not propagate over negative links in SNs \cite{Tang_2016}, a susceptible user gets infected with rate $\beta\!\in\!\mathbb{R}^+$\,times the number of its infected friendly contacts~\cite{Faryad_2011}. A user aware of the process however, is less~likely to get infected, with a lower infection rate $0 \!\leq\! \beta_a \!<\! \beta$, as compared to a susceptible user. Unlike the irreversible SI model in \cite{Saeedian_2017}, a susceptible user becomes aware of the process with rate $\kappa \!\in\! \mathbb{R}^+$ times the number of direct infected friends and all infected users may eventually recover back to susceptibility with rate $\delta \!\in\! \mathbb{R}^+$. For all $i \!\in\! \mathcal{V}$, the network state can thus, be expressed formally as the CTMC $\{X_i(t); t \!\geq\! 0\}$, where:\vspace{-0.5em}
\begin{equation}
	\label{eq1}
	X_i(t) =
	\begin{cases} 
      	1; & \text{if user $i$ is susceptible at time $t$},\vspace{-0.2em} \\
      	%\noalign{\vskip0.1pt}
      	0; & \text{if user $i$ is alert at time $t$},\vspace{-0.2em} \\
      	%\noalign{\vskip0.1pt}
      	-1; & \text{if user $i$ is infected at time $t$}.
   	\end{cases}
   	\vspace{-0.5em}
\end{equation}
Using \eqref{eq1}, we now can define the probability of user $i$ being in one of the three epidemic states as $S_i(t) = \Pr[X_i(t)=1]$, $A_i(t)\! =\! \Pr[X_i(t)=0]$, and $I_i(t) \! = \! \Pr[X_i(t) \!=\! -1]$ such that for $1 \leq i \leq n$, $S_i(t) + A_i(t) + I_i(t) = 1$ always holds.\vspace{-0.6em}
\begin{figure}[!t]
	\centering
	\includegraphics[width=3.05in]{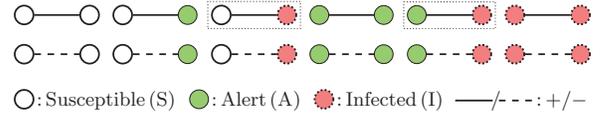}
	\vspace{-0.8em}
	\caption{All possible pairwise links in the SAIS spreading model. The edges enclosed in boxes are the only two transmissible configurations.}
	\label{fig2}
	\vspace{-0.4em}
\end{figure}
%\footnote{Henceforth, we drop the time parameter $t$ to enhance readability.}

\subsection{Pairwise Spreading Energy Function}
\label{sec2.2}
\fontdimen2\font=0.43ex
We now delineate the sign evolution of user interactions~in~the context of energy. Given the three epidemic states (S, A, I) and binary link signs ($-,+$), there exist $12$ distinct user pair configurations as shown in \figurename{~\ref{fig2}}. We characterize the viral potency by mapping each user pair configuration $(i,j)$, where $i,j \in \mathcal{V}$, to the energy landscape as follows:\vspace{-0.5em}
\begin{equation}
	\label{eq2}
	E^{p}_{i,j}(t) \!\triangleq \!
	\begin{cases} 
      	\!\mathcal{A}_{i,j}(t) \frac{\big(\!X_i(t) - X_j(t)\!\big)^2}{4}; & \!\!\text{if $\lvert X_i \!+\! X_j \rvert \!\bmod\! 2 \!=\! 0$},\! \\
      	\noalign{\vskip4pt}
      	\!\mathcal{A}_{i,j}(t) \frac{1 - X_i(t) - X_j(t)}{2}; & \!\!\text{otherwise}.
   	\end{cases}
   	\vspace{-0.5em}
\end{equation}

Based on the functional value of \eqref{eq2}, the configurations depicted in \figurename{~\ref{fig2}} can be classified as follows:\vspace{-0.4em}
\begin{itemize}
 	\item  \emph{Balanced edges:} As long as configurations $\mathrm{S}-\mathrm{I}$ and A$-$I do not flip their edge signs while evolving, the users~$i$ and $j$ are in a balanced social relationship and do~not engage in the propagation process \cite{Saeedian_2017}. Therefore,~they exhibit a pairwise energy of $E^{p}_{i,j}(t)=-1$.
 	\item \emph{Unbalanced edges:} Cases in which a susceptible or alert user is in a friendly relationship with an infected user are socially unstable and are bound to change with time. In our model, $\mathrm{S}+\mathrm{I}$ and $\mathrm{A}+\mathrm{I}$ serve as feasible links for epidemic spread and thus, the users are in an unbalanced state with pairwise energy of $E^{p}_{i,j}(t)=1$.
 	\item \emph{Neutral edges:} Irrespective of the edge sign, configurations $\mathrm{S}\pm\mathrm{S}$, $\mathrm{I}\pm\mathrm{I}$, $\mathrm{A}\pm\mathrm{A}$, and $\mathrm{S}\pm\mathrm{A}$ do not contribute to the spreading process, and thus, exhibit zero pairwise energy, i.e., $E^{p}_{i,j}(t)=0$.
\end{itemize}\vspace{-0.2em}

Accordingly, the total pairwise spreading energy of network $\mathcal{G}_t$, denoted by $E_{p}(\mathcal{G}_t)$, can be computed as:\vspace{-0.4em}
\begin{equation}
	\label{eq3}
	E_{p}(\mathcal{G}_t) = \frac{1}{\binom{n}{2}}\sum\limits_{\substack{i,j \\ i \neq j}} E^{p}_{i,j}(t).
	\vspace{-0.4em}
\end{equation}

\subsection{Triad Structural Energy Function}
\label{sec2.3}
\fontdimen2\font=0.45ex
Along with the users' epidemic states, edge sign evolution is also driven by Heider's structural balance criterion. A triad~of users in $\mathcal{G}_t$, denoted by $(i,j,k)$, is said to be balanced if the product of $\mathcal{A}_{i,j}(t) \cdot \mathcal{A}_{j,k}(t) \cdot \mathcal{A}_{k,i}(t)$ is positive. In other words, an unbalanced triad will always have an odd number of~negative edges. Hence, network $\mathcal{G}_t$ is \textit{fully} balanced only if all the constituent triads are balanced. Conforming to Heider's balance theory, the structural status of any triad $(i,j,k)$ can~be mapped to the energy landscape as below \cite{Heider_1958}:\vspace{-0.2em}
\begin{equation}
	\label{eq4}
	E^{\blacktriangle}_{i,j,k}(t) \!\triangleq \!-\mathcal{A}_{i,j}(t) \cdot \mathcal{A}_{j,k}(t) \cdot \mathcal{A}_{k,i}(t).
	\vspace{-0.3em}
\end{equation}

With \eqref{eq4} in place, the energy contribution of any balanced (unbalanced) triad in $\mathcal{G}_t$ is $E^{\blacktriangle}_{i,j,k}(t) \!=\! -1$ $(E^{\blacktriangle}_{i,j,k}(t) \!=\! 1)$.~\figurename{\,\ref{fig3}} showcases all the balanced triads for the SAIS and~the SIS (baseline) models. To converge towards lower energy states (which corresponds to more social stability), users in unbalanced triads tend to flip their link signs which in turn, affects the configuration of other triads that share common edges with them. Consequently, the total normalized energy of $\mathcal{G}_t$, denoted by $E_{\triangle}(\mathcal{G}_t)$, is as follows, where $-1 \!\leq\! E_{\triangle}(\mathcal{G}_t) \!\leq\! 1$:\vspace{-0.4em}
\begin{equation}
	\label{eq5}
	E_{\triangle}(\mathcal{G}_t) = \frac{1}{\binom{n}{3}}\sum\limits_{\substack{i,j,k \\ i \neq j \neq k}} E^{\blacktriangle}_{i,j,k}(t)\,,
	\vspace{-0.3em}
\end{equation}
Users in triads decide on whether or not to alter their relationships only if the total triad energy of the resulting network is further reduced. Apparently, SNs that manifest triad energy values closer to $-1$ tend to be socially more stable and thus, pragmatically justified.\vspace{-0.8em}
\begin{figure}[!t]
	\centering
	\includegraphics[width=2.5in]{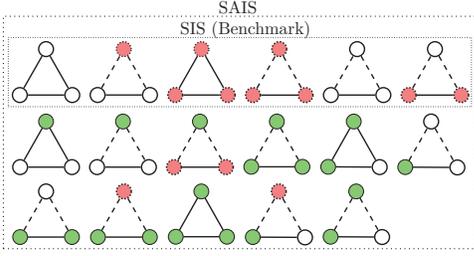}
	\vspace{-0.9em}
	\caption{All possible balanced triads in the SAIS and the SIS models.}
	\label{fig3}
	\vspace{-0.4em}
\end{figure}

\subsection{Weighted Network Energy Function}
\label{sec2.4}
\fontdimen2\font=0.45ex
We now define the total energy of network $\mathcal{G}_t$, given by $E(\mathcal{G}_t)$, as the weighted sum of the overall pairwise and triad energy functions derived in \eqref{eq3} and \eqref{eq5}, respectively:\vspace{-0.3em}
\begin{equation}
	\label{eq6}
	E(\mathcal{G}_t) = \alpha \cdot E_{\triangle}(\mathcal{G}_t) + (1 - \alpha) \cdot E_p(\mathcal{G}_t)\,,
	\vspace{-0.3em}
\end{equation}
where $\alpha (0 \!\leq\! \alpha \!\leq\! 1)$ is the tuning parameter used to adjust~the energy trade-off between the epidemic spread $(\alpha \!=\! 0)$ and the structural balance $(\alpha \!=\! 1)$ in the network. Hence, if~$\mathcal{G}_t$ is fully balanced (fully unbalanced), then $E(\mathcal{G}_t) \!=\! -1$ $(E(\mathcal{G}_t) \!=\! 1)$, which is more likely to be achieved in smaller graphs. Note that for some fixed $\alpha$ value, attaining the global~(local) energy minimum state in which $\mathcal{G}_t$ is fully (or nearly) balanced~depends on the initial fractions of infected users $(0 \leq \rho_0 \leq 1)$ and friendly links $(0 \leq r_0 \leq 1)$ \cite{Saeedian_2017}.\vspace{-0.8em}
\begin{figure}[!t]
	\centering
	\includegraphics[width=3.2in]{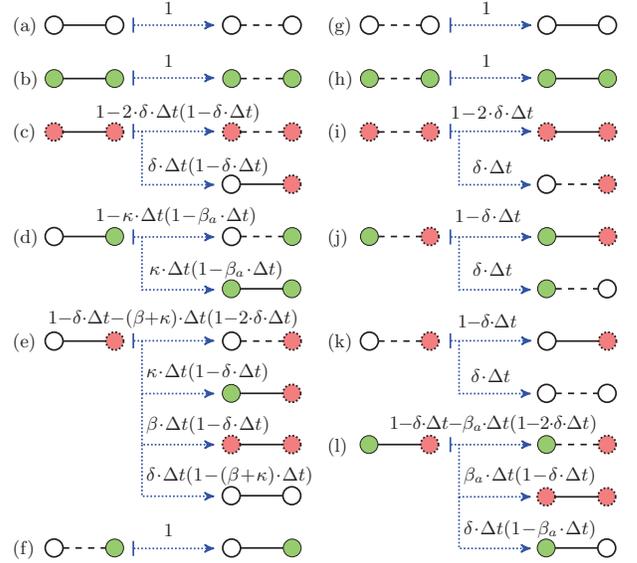}
	\vspace{-0.6em}
	\caption{Transition probabilities for temporal evolution of states in \figurename{\,\ref{fig2}}.}
	\label{fig4}
	\vspace{-0.4em}
\end{figure}

%Next, we determine the conditional transition probability rule set that governs the temporal evolution of the user pair configurations given in \figurename{\,\ref{fig2}}, and then propose our CTMC.

\section{Steady-State Probability Distribution}
\label{sec3}
\fontdimen2\font=0.44ex
Assuming that only one event is triggered in each time step, $\Delta t \!\ll\! 1$, i.e., either the epidemic state of exactly one user pair $(i,j)$ changes or the edge sign is flipped, the rules defining the user pair transitions are shown in \figurename{~\ref{fig4}}. The balanced edges $\mathrm{S} \!-\! \mathrm{I}$ and $\mathrm{A} \!-\! \mathrm{I}$ transition to $\mathrm{S} \!-\! \mathrm{S}$ and $\mathrm{A} \!-\! \mathrm{S}$, respectively, with probability $\delta \cdot \Delta t$ or change to $\mathrm{S} \!+\! \mathrm{I}$ and $\mathrm{A} \!+\! \mathrm{I}$, respectively, with probability $1 \!-\! \delta \cdot \Delta t$. For unbalanced edges, $\mathrm{S} \!+\! \mathrm{I}$ changes to $\mathrm{S} \!-\! \mathrm{I}$ with probability $1 \!- \delta \cdot \Delta t - (\beta + \kappa) \cdot \Delta t \cdot (1 - 2 \delta \cdot \Delta t)$ or switches to states $\mathrm{A} \!+\! \mathrm{I}$, $\mathrm{I} \!+\! \mathrm{I}$, or $\mathrm{S} \!+\! \mathrm{S}$ with probabilities $\kappa \cdot \Delta t (1 - \delta \cdot \Delta t)$, $\beta \cdot \Delta t (1 - \delta \cdot \Delta t)$, and $\delta \cdot \Delta t (1 - (\beta + \kappa) \cdot \Delta t)$, respectively.~Also, $\mathrm{A} \!+\! \mathrm{I}$ switches to $\mathrm{A} \!-\! \mathrm{I}$, $\mathrm{I} \!+\! \mathrm{I}$, or $\mathrm{A} \!+\! \mathrm{I}$ with probabilities $1 \!- \delta \cdot \Delta t - \beta_a \cdot \Delta t \cdot (1 - 2 \delta \cdot \Delta t)$, $\beta_a \cdot \Delta t (1 - \delta \cdot \Delta t)$, and $\delta \cdot \Delta t (1 - \beta_a \cdot \Delta t)$, respectively. Among the neutral edges, $\mathrm{S} \!\pm\! \mathrm{S}$, $\mathrm{A} \!\pm\! \mathrm{A}$, and $\mathrm{S} \!-\! \mathrm{A}$ flip their edge signs with probability $1$ in each time step, $\mathrm{S} \!+\! \mathrm{A}$ changes to either $\mathrm{S} \!-\! \mathrm{A}$ or $\mathrm{A} \!+\! \mathrm{A}$ with probabilities $1 - \kappa \cdot \Delta t \cdot (1 - \beta_a \cdot \Delta t)$ and $\kappa \cdot \Delta t \cdot (1 - \beta_a \cdot \Delta t)$, respectively. Finally, $\mathrm{I} \!\pm\! \mathrm{I}$ changes to $\mathrm{I} \!\mp\! \mathrm{I}$ with probabilities $1 - 2\delta \cdot \Delta t \cdot (1 - \delta \cdot \Delta t)$ and $1 - 2 \delta \cdot \Delta t$, respectively, and to $\mathrm{S} \!\pm\! \mathrm{I}$ with probabilities $\delta \cdot \Delta t \cdot (1 - \delta \cdot \Delta t)$ and $\delta \cdot \Delta t$, respectively.

In general, given the pair $(i,j)$, the tri-variate~CTMC of the form $\{Z_{i,j}(t); t \!\geq\! 0\}$, where $Z_{i,j}(t) \!=\! \big(X_i(t), X_j(t), \mathcal{A}_{i,j}(t)\big)$, defines these conditional transition probabilities as follows:\vspace{-0.3em}
\begin{align} 
	\label{eq7}
		\mathcal{P}_{c,c'}(\Delta t) \triangleq \Pr \Big[Z_{i,j}(t + \Delta t) = c' \big| Z_{i,j}(t) = c\Big],
		\vspace{-0.3em}
\end{align}
where $c \!=\! (x,y,z)$, $c' \!=\! (x',y',z')$, and $x,y,z,x',y',z' \!\in \{-1,0,1\}$. Based on \eqref{eq7}, the stead-state probability distribution is derived in Theorem~1.
\begin{theorem}
Let $\big\{\pi_{x,y,z} \big| x,y,z \in \{-1,0,1\}\big\}$ be the stationary probabilities for the CTMC $Z_{i,j}(t)$ defined above, then we have the following in steady-state:\vspace{-0.3em}
	\begin{enumerate}
  		\item[(i)] The fraction of susceptible users $(s_{\infty})$ is $\sum_{y,z} \pi_{1,y,z}$.
  		\item[(ii)] The fraction of infected users $(\rho_{\infty})$ is $\sum_{y,z} \pi_{-1,y,z}$.
  		\item[(iii)] The fraction of alerted users $(a_{\infty})$ is $\sum_{y,z} \pi_{0,y,z}$.
  		\item[(iv)] The fraction of friendly links $(r_{\infty})$ is $\sum_{x,y} \pi_{x,y,1}$.
\end{enumerate}
\end{theorem}
\begin{proof}
	We use \eqref{eq7} to obtain the elements of the infinitesimal generator matrix $Q=[q_{c,c'}]$ of order $27$ as follows:\vspace{-0.3em}
	\begin{equation}
	\label{eq8}
	q_{c,c'} = 
	\begin{cases} 
      	\lim\limits_{\Delta t \to 0} \frac{\mathcal{P}_{c,c'}(\Delta t) - 1}{\Delta t}; & \text{if $c' = c$}, \\
      	\noalign{\vskip2pt}
      	\lim\limits_{\Delta t \to 0} \mathcal{P}_{c,c'}(\Delta t); & \text{if $c' \neq c$}.
   	\end{cases}
   	\vspace{-0.3em}
	\end{equation}
	From \eqref{eq8}, we now can obtain the stationary probabilities by solving $\Pi \cdot Q = 0$ and $\Pi \cdot 1 = 1$, where $\Pi = \big\{\pi_{x,y,z} \big| x,y,z \in \{-1,0,1\}\big\}$. Denoted by $s_{\infty} \!=\! \sum_{i=1}^n S_i(\infty) / |\mathcal{V}|$, the fraction of susceptible users is computed as:\vspace{-0.4em}
	\begin{align}
		s_{\infty} 
%		=\, &\pi_{1,1,1} + \pi_{1,1,0} + \pi_{1,1,-1} + \pi_{1,0,1} + \pi_{1,0,0} + \pi_{1,0,-1} \nonumber \\
%		               & + \pi_{1,-1,1} + \pi_{1,-1,0} + \pi_{1,-1,-1} \nonumber \\
		                         =\, &\sum_{y,z} \pi_{1,y,z}.  \nonumber\vspace{-0.5em}
	\end{align}
Similarly, the steady-state probabilities for $\rho_{\infty}$, $a_{\infty}$, and $r_{\infty}$ can be obtained straightforwardly.
\end{proof}
\begin{algorithm}[!t]
 \caption{Time-efficient Network Energy Calculation}
 {\small\begin{algorithmic}[1]
 \renewcommand{\algorithmicrequire}{\textbf{Input:}}
 \renewcommand{\algorithmicensure}{\textbf{Output:}}
 \REQUIRE $\mathcal{G}_0 \!=\! (\mathcal{V},\mathcal{E}_0)$, $\alpha$, $\rho_0$, $r_0$, $\beta$, $\beta_a$, $\kappa$, and $\delta$.
 \ENSURE   $E_{min}(\mathcal{G})$.
 \\ \textit{Initialization}: $\forall (i,j), \mathcal{A}_{i,j}(0) \in \mathcal{E}_0=-1$, and $e = 1$.%, and $s = \textsc{false}$.
 \STATE Calculate $E(\mathcal{G}_0)$ using \eqref{eq6}.
 \STATE $e \gets E(\mathcal{G}_0)$.
 %\WHILE {$s == \textsc{False}$}
	\FOR {$t \gets \Delta t$ to $T \cdot \Delta t$}
		\STATE Randomly select an edge $(i,j) \in \mathcal{E}_{t-\Delta t}$.
		\STATE Change the state of edge $(i,j)$ according to \figurename{\,\ref{fig4}}.
		\STATE Update $E(\mathcal{G}_t)$ using \eqref{eq9}, \eqref{eq11}, and \eqref{eq12}.
		\IF {$E(\mathcal{G}_t) == e$}
			\STATE $e \leftarrow E(\mathcal{G}_t)$ with probability $0.5$.
		\ELSIF {$E(\mathcal{G}_t) < e$}
			\STATE $e \gets E(\mathcal{G}_t)$
		\ENDIF
	\ENDFOR
	\RETURN $E_{min}(\mathcal{G}) \leftarrow e$\vspace{-0.2em}
 \end{algorithmic} }
 \label{alg1}
\end{algorithm}
\vspace{-1em}

\section{Monte Carlo Method}
\label{sec4}
\fontdimen2\font=0.42ex
Starting\,\,from\,\,an\,\,initial\,\,network\,\,state\,\,at\,\,$t$\,$=$\,$0$,\,where a fraction of users $(\rho_0)$ are randomly infected, we select an edge $(i,j)$ at random in each evolution step of the simulation and change its state as in \figurename{~\ref{fig4}}. Doing so affects the energy states of all triads that share edge $(i,j)$ in the long-term~which successively, alters the total network energy state.~Convergence towards the new network structure transpires as long as the new energy state decreases in each time step \cite{Saeedian_2017}, i.e., this process continues until the global minimum $(E(\mathcal{G}_t) \!=\! -1)$ or a local $(E(\mathcal{G}_t) \!>\! -1)$ minimum energy state is reached. Computing $E(\mathcal{G}_t)$ using \eqref{eq6} in each time step takes $O\big(\binom{n}{2} + \binom{n}{3}\big)$ time. To expedite the computation, we propose Algorithm~\ref{alg1} that evaluates the energy\,\,difference\,\,of the selected~edge between consecutive time steps, $t'$ and $t''$ $(t'' \!=\! t' \!+\! \Delta t)$, in $O(1)$ time as:\vspace{-0.4em}
\begin{equation}
	\label{eq9}
	\Delta E(i,j) = \alpha \cdot \Delta E_{\triangle}(i,j) + (1 - \alpha) \cdot \Delta E_p(i,j)\, ,
	\vspace{-0.4em}
\end{equation}
where $\Delta E_{\triangle}(i,j)$ is the difference in the triad energy, i.e.,\vspace{-0.4em}
\begin{align}
	\label{eq10}
	\!\!\!\!\Delta E_{\!\triangle\!}(i,\!j) \!&=\! \frac{1}{\binom{n}{3}}\!\sum_{i'\!,j'\!,k'} \!\!\!\big(E^{\blacktriangle}_{i'\!,j'\!,k'\!}(t'') \!-\! E^{\blacktriangle}_{i'\!,j'\!,k'\!}(t')\big)\nonumber \vspace{-0.1em}\\
	\!&=\! \frac{\mathcal{A}_{i,j}(t')\! - \! \mathcal{A}_{i,j}(t'')}{\binom{n}{3}}\!\!\sum_{k'\neq i,j}\!\! \mathcal{A}_{i,k'\!}(t') \cdot \mathcal{A}_{j,k'\!}(t').
	\vspace{-1.5em}
\end{align}
If the state transition does not flip the edge sign, then $E_{\triangle}(\mathcal{G}_t)$ remains unaltered. Otherwise, flipping the edge sign implies that $\mathcal{A}_{i,j}(t'') \!= \!- \mathcal{A}_{i,j}(t')$, which further reduces \eqref{eq10} to:\vspace{-0.4em}
\begin{align}
	\label{eq11}
%	\Delta &E_{\!\triangle\!}(i,j) = \quad\nonumber \\
	\!\!&\begin{cases} 
      	\!\!\frac{-2 \mathcal{A}_{i,j}(t'')}{\binom{n}{3}} \!\!\!\sum\limits_{k'\neq i,j}\!\!\!\!\mathcal{A}_{i,k'\!}(t') \!\cdot\! \mathcal{A}_{j,k'\!}(t'); & \text{if $\mathcal{A}_{i,j\!}(t'') \!=\! -\mathcal{A}_{i,j\!}(t')$},\vspace{-0.3em} \\
      	\noalign{\vskip0.3pt}
      	0; & \text{if $\mathcal{A}_{i,j}(t'') \!=\! \mathcal{A}_{i,j}(t')$}.
      	\vspace{-0.3em}
   	\end{cases}
   	\vspace{-0.7em}
\end{align}
Similarly, $\Delta E_p(i,j)$ can be computed as follows:\vspace{-0.5em}
\begin{equation}
	\label{eq12}
	\Delta E_p(i,j) = \frac{1}{\binom{n}{2}}\big(E^p_{i,j}(t'') - E^p_{i,j}(t')\big)\,.
	\vspace{-0.7em}
\end{equation}

\section{Simulation Results and Discussions}
\label{sec5}
\fontdimen2\font=0.42ex
We evaluate the proposed energy~framework with respect~to $r_0$, $\rho_0$, and $\alpha$. For Case Study I, we generate a complete network of $|\mathcal{V}| \!=\! 180$ users to entail the maximum number of triads in our analysis. We then apply the bootstrap technique to extract i.i.d. samples of smaller connected components of size $|\mathcal{V}|$ from the Slashdot$081106$ (SL, $|\mathcal{V}| \!=\! 747$)\,\cite{snapnets} and~the Bitcoin-OTC (BC, $|\mathcal{V}| \!=\! 709$)\,\cite{snapnets} datasets, and the entire~US Congress co-sponsorship (CS, $|\mathcal{V}| \!=\! 100$, $|\mathcal{E}_t| \!=\! 3696$)\,\cite{Neal2020} dataset in Case Study II. To ensure better theoretical predictions, all Monte Carlo simulation results are averaged over $100$ runs~on a PC with $\SI{3.2}{\giga\hertz}$ Intel Core i$9$-$9900$ CPU and $\SI{16}{\giga\byte}$ RAM\footnote{GitHub repository: https://github.com/cnsl-nu/Co-evolution-of-Viral-Processes-and-Structural-Stability-in-Signed-Social-Networks}.\vspace{-0.5em}
\begin{figure}[!t]
\centering
\subfloat[\vspace{-0.1em}Susceptible density vs. time]{\label{fig5a}\includegraphics[width=1.564in]{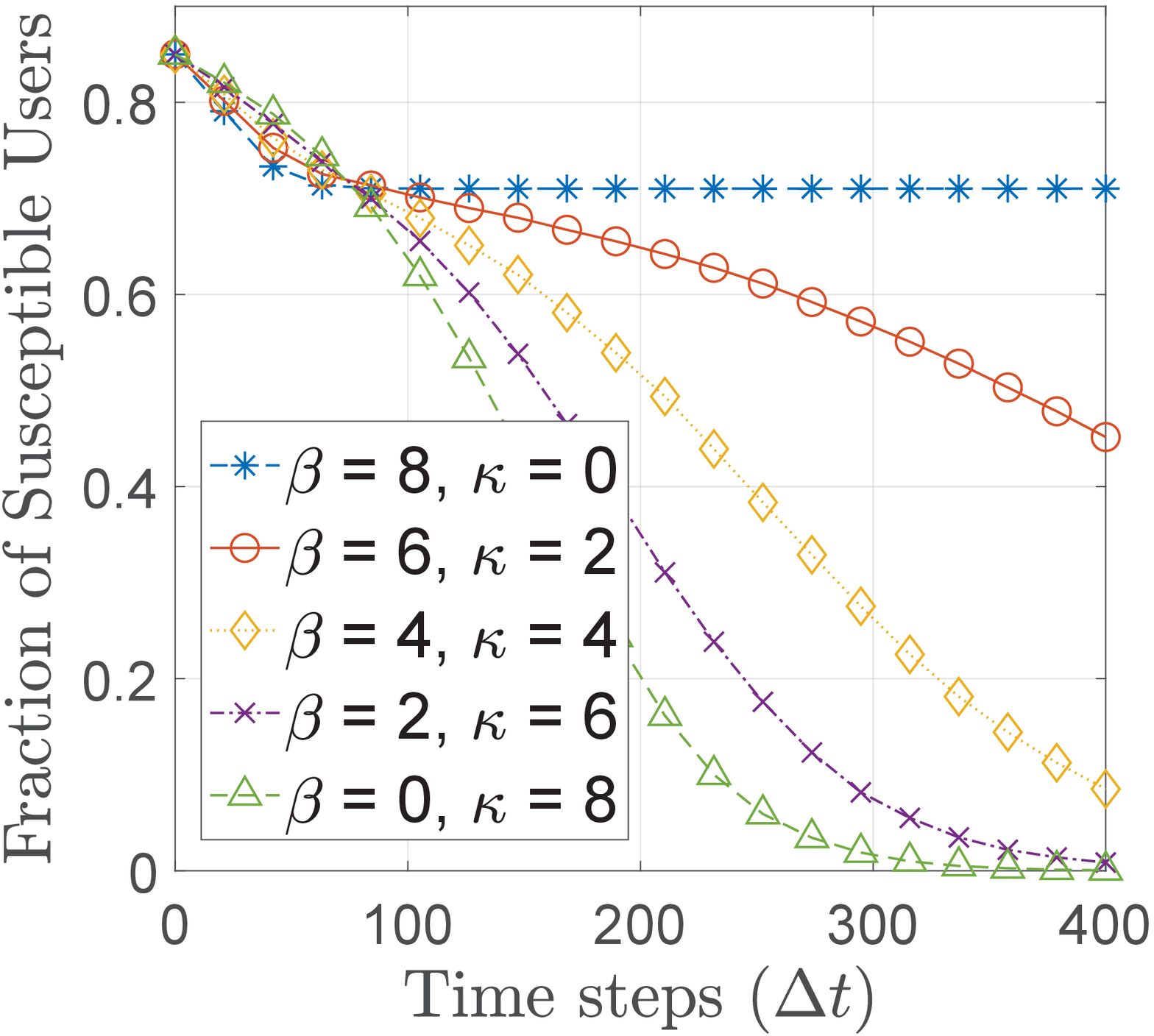}}
~~
\subfloat[\vspace{-0.1em}Alerted density vs. time]{\label{fig5b}\includegraphics[width=1.565in]{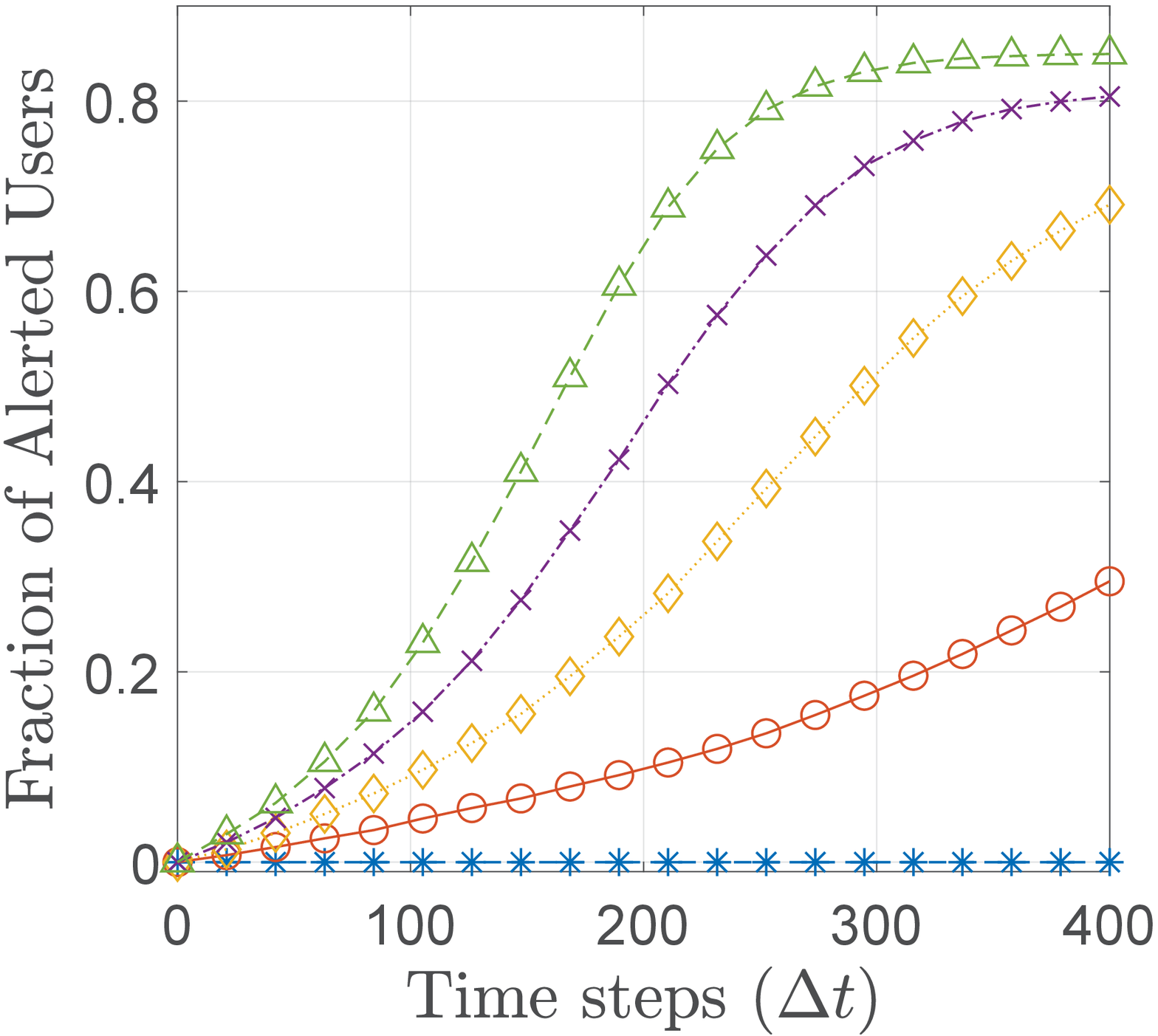}}
\vspace{-0.4em}

\centering
\subfloat[\vspace{-0.1em}Infected density vs. time]{\label{fig5c}\includegraphics[width=1.566in]{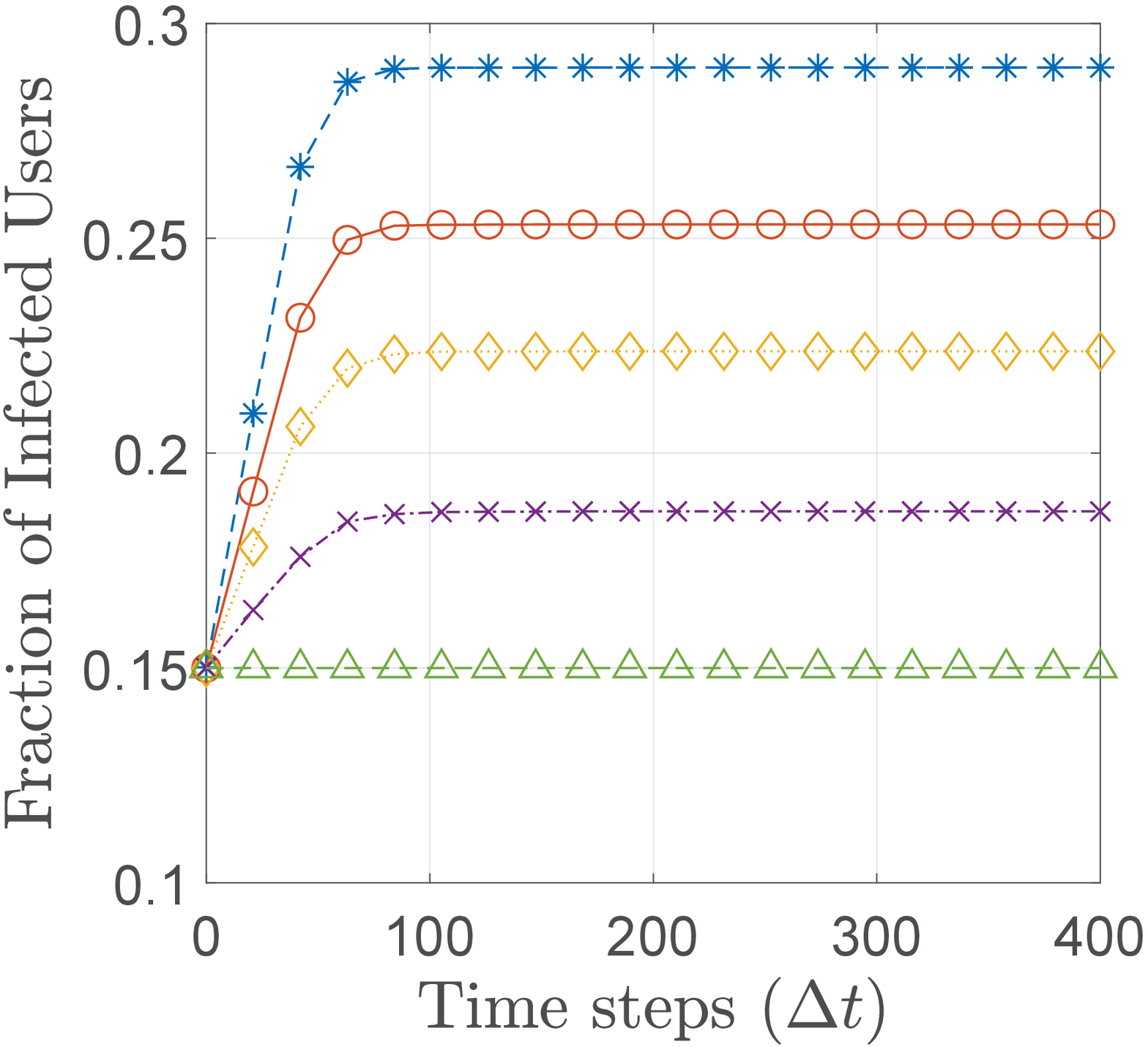}}
~~
\subfloat[\vspace{-0.1em}No. of freindly links vs. time]{\label{fig5d}\includegraphics[width=1.55in]{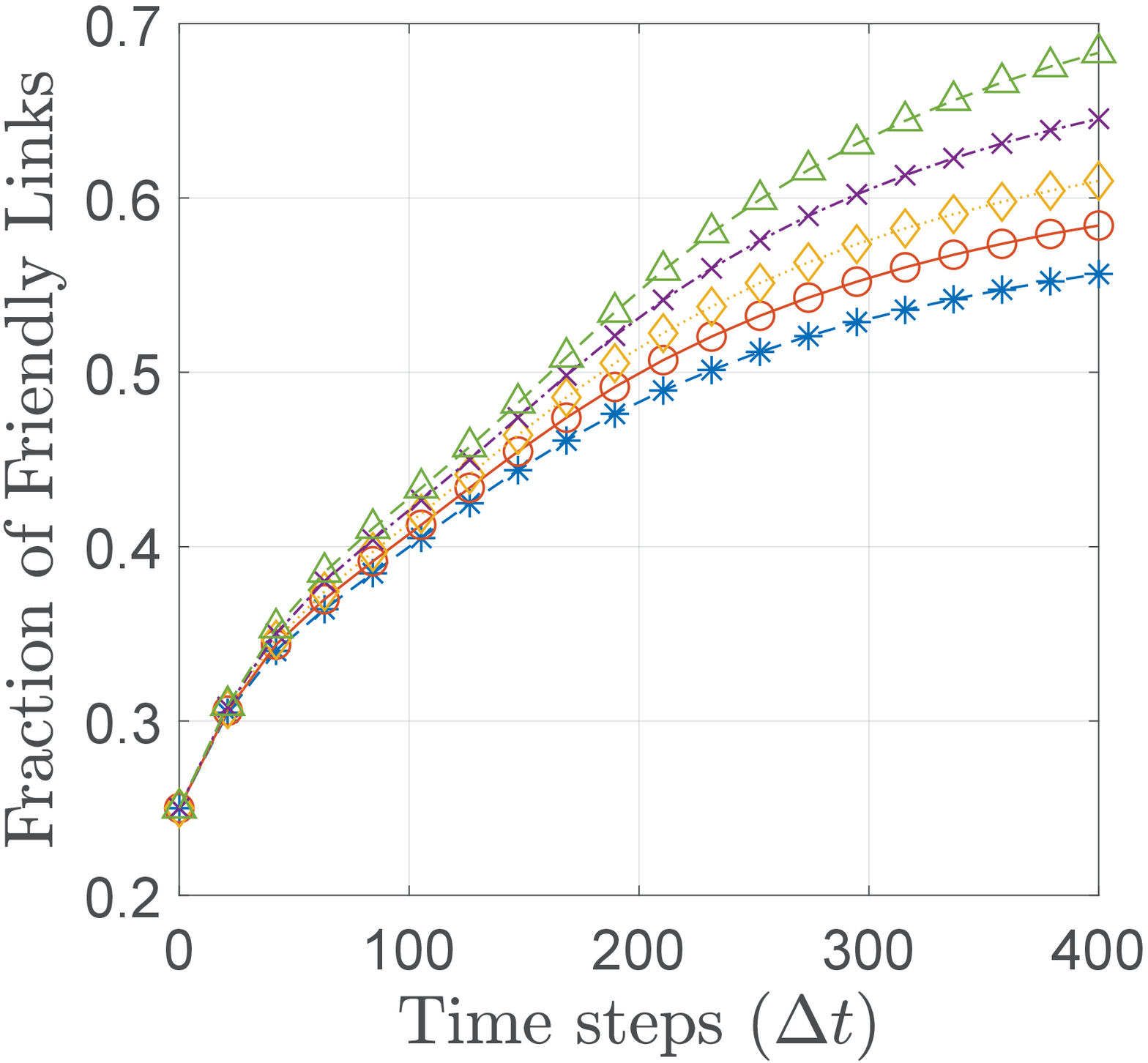}}
\vspace{-0.3em}
\caption{Time evolution of the proposed spreading model for different $(\beta, \kappa)$ values. The SIS baseline model is traced as ${(\beta=8, \kappa=0)}$.}
\vspace{-0.5em}
\label{fig5}
\end{figure}

\subsection{Case Study I: Results}
\label{sec5.2}
\fontdimen2\font=0.42ex
For the synthetic network, \figurename{~\ref{fig5}} plots the trajectories for a~viral outbreak under varying $(\beta,\kappa)$ values with $r_0\!=\!0.25$, $p_0 \!=\! 0.15$,~$\alpha \!=\! 0.5$,~$\beta_a \!=\! 0.3 \beta$,~and~$\delta \!=\! 9$ \cite{Faryad_2011, Sahneh_2019}. As shown in \figurename{~\ref{fig5a}}, unlike the SIS baseline ${(\kappa \!=\! 0)}$, the susceptible~fraction under the SAIS model drops to zero with rate proportional~to $\kappa$.\,As\,a\,result,\,the\,\,users\,\,are\,\,either\,\,influenced\,\,by\,the viral~process or alerted thus making them less likely to fall prey in the long run since $\beta_a \!\ll\! \beta$. The impact of $\kappa$ on the alerted\,and\,infected\,user\,densities\,are, respectively, shown~in \figurename{~\ref{fig5b}} and \figurename{~\ref{fig5c}}. For\,higher\,$\kappa$\,values, a larger~susceptible fraction\,is\,made\,aware of the spread which, in turn, diminishes the size of the infected population. For instance,~in~contrast to the baseline, \figurename{~\ref{fig5c}} vividly shows that the infected cluster size decreases by nearly $14\%$ when $\kappa \!=\! 2$. In spite of setting $\beta \!=\! 0$, note that there exists a non-zero infected population in the network. This clearly implies that for smaller values of $r_0$, the proposed model partitions the network into two distinct clusters: one comprising of alerted users and the other containing infected users. While the users within each cluster maintain a friendly relationship with each other, they are hostile towards users in the other cluster. The virality~of the process however, dies out gradually with rise in $r_0$ and the two clusters eventually merge into a single cluster of alert users. The impact of $\kappa$ on the number of friendly links is~shown in \figurename{~\ref{fig5d}}. Driven by the transitions given in \figurename{~\ref{fig4}}, such behavior is not far from expectation as users are inclined to detach from friends influenced by the spreading phenomena and instead, befriend those who are informed or share common interests to attain social stability.

\figurename{~\ref{fig6}} shows how the control parameter $\alpha$ arbitrates the epidemic spread and social tension trade-off for varying $r_0$ values in steady-state. In \figurename{~\ref{fig6a}}, we observe that the system gravitates towards the jammed states $(E(\mathcal{G}) \!<\! -1)$, where mitigating the epidemic is favored over attaining~structural balance, for lower $(\alpha, r_0)$ values. As $\alpha$ and $r_0$ goes beyond $0.5$ however, it is evident in \figurename{~\ref{fig6b}} that the network tends towards the global minimum energy state to become structurally robust at the expense of further epidemic spread. Interestingly, due to the reversible nature of the SAIS model, the fully balanced complete network progresses to be infection-free at $\alpha \!=\! r_0 \!=\! 1$ as all users eventually become aware~of the spread. Hence, the network tends to exploit the negative links to naturally immunize the susceptible users by separating them from the cluster of infected users. Therefore, for any~given setting, an optimal $(\alpha, r_0)$ pair exists for which~the network would contain minimum number of infected users in steady-state and yet, not necessarily be socially balanced.
\begin{figure}[!t]
\centering
\subfloat[$\rho_{\infty}$ vs. $(\alpha,r_0)$]{\label{fig6a}\includegraphics[width=1.52in]{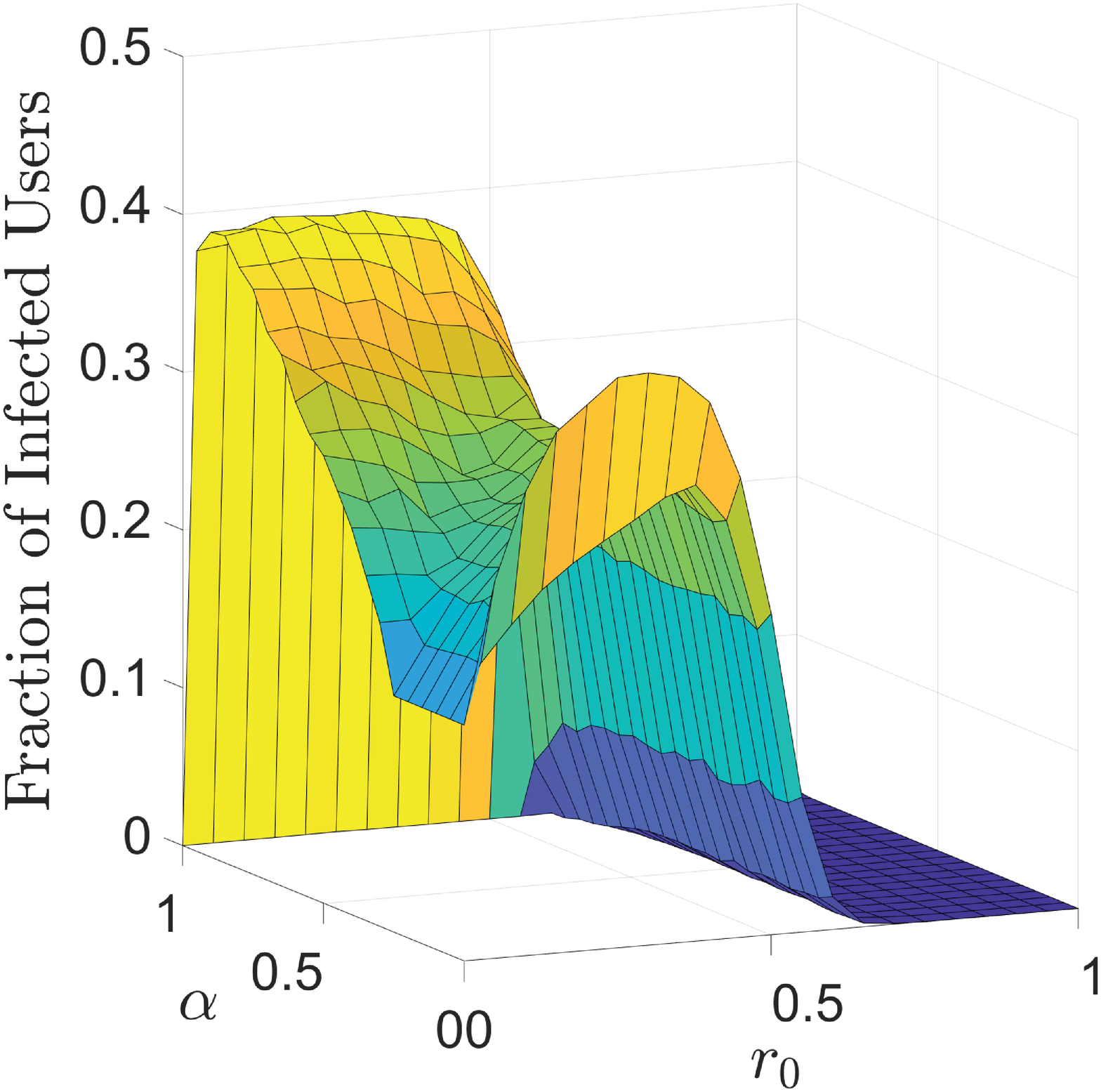}}
\,\,\,
\subfloat[$E(\mathcal{G})$ vs. $(\alpha,r_0)$]{\label{fig6b}\includegraphics[width=1.62in]{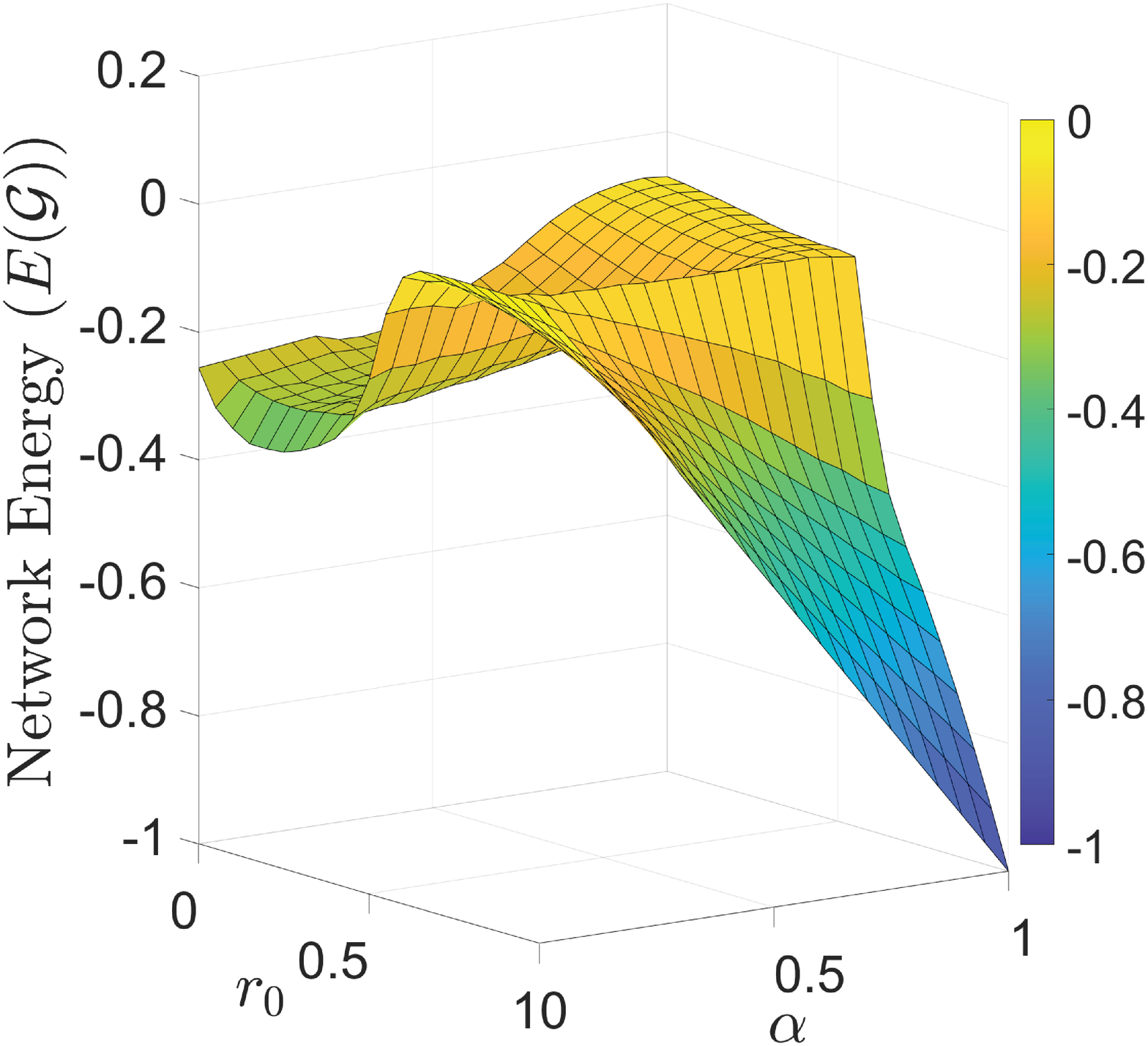}}
\vspace{-0.4em}
\caption{Impact of $\alpha$ and $r_0$ on the steady-state infection density and the network energy for $p_0 \!=\! 0.15$, $\beta \!=\! 6$, $\kappa \!=\! 4$, $\beta_a \!=\! 0.3 \beta$, and $\delta \!=\! 9$.}
\label{fig6}
\vspace{-0.7em}
\end{figure}
\begin{figure}[t]
\centering
\subfloat[$\rho_{\infty}$ vs. $(\rho_0,r_0)$]{\label{fig7a}\includegraphics[width=1.6in]{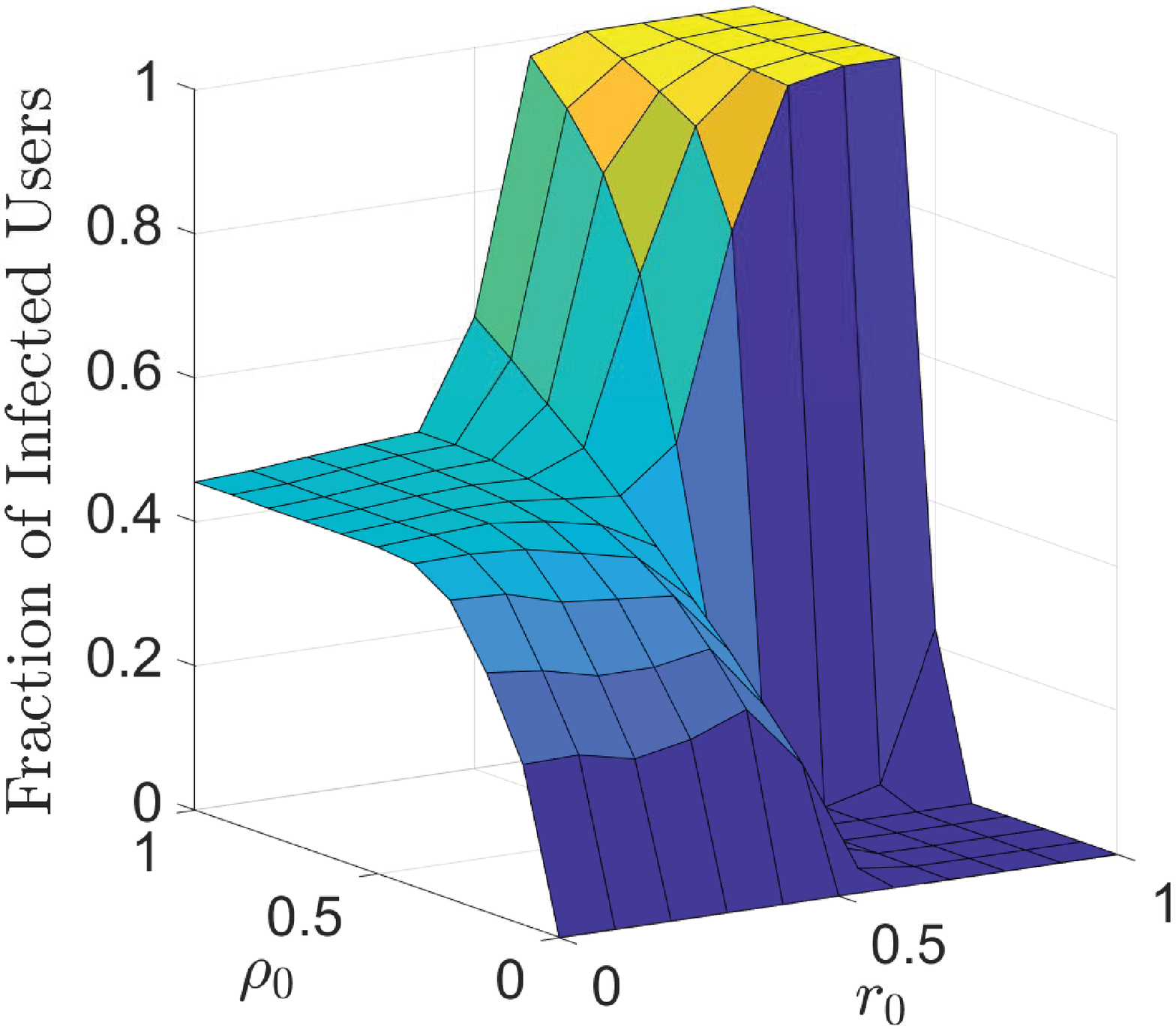}}
\,
\subfloat[$E(\mathcal{G})$ vs. $(\rho_0,r_0)$]{\label{fig7b}\includegraphics[width=1.69in]{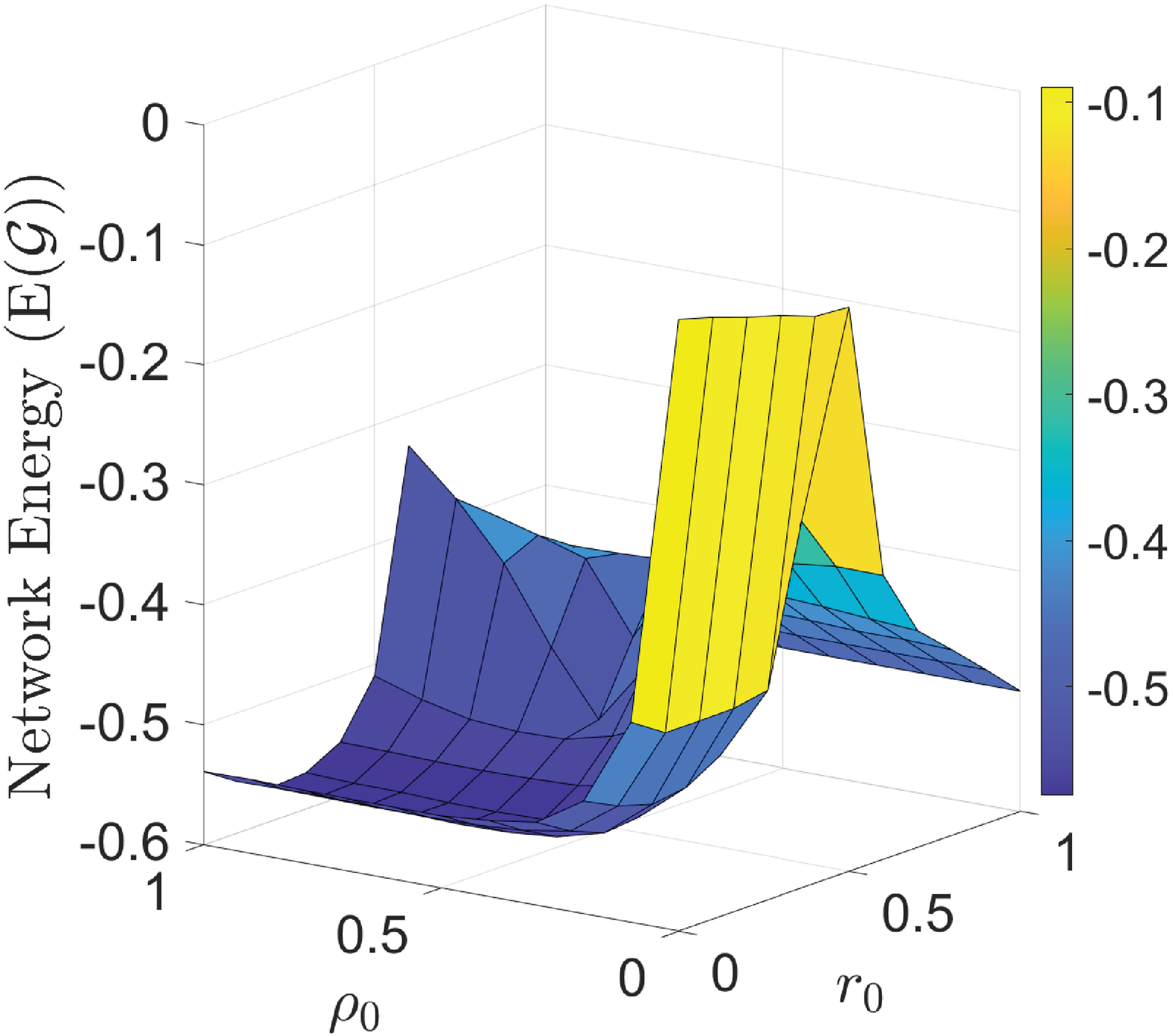}}
\vspace{-0.5em}
\caption{Impact of $\rho_0$ and $r_0$ on the steady-state infection density and the network energy for $\alpha \!=\! 0.5$, $\beta \!=\! 6$, $\kappa \!=\! 4$, $\beta_a \!=\! 0.3 \beta$, and $\delta \!=\! 9$.}
\vspace{-1em}
\label{fig7}
%\vspace{-1.0em}
\end{figure}

The significance of $\rho_0$ in the coupled evolution is shown in \figurename{~\ref{fig7}}. For $0 \!\leq\! \rho_0, r_0 \!\leq\! 0.5$, \figurename{~\ref{fig7a}} showcases the impact of negative links in controlling the contagious spread. Further increase in $\rho_0$ however, yields a fixed fraction of infected~users as most of the triads have evolved into a balanced state.~Full recovery is attained when $\rho_0 \!\leq\! 0.5$ and $r_0 \!>\! 0.5$ due to the small infection prevalence and the low infection rate relative to $\kappa$ and $\delta$. Also, note that for low $\rho_0$, the stationary infection density ($\rho_{\infty}$) drops to zero, irrespective of $r_0$. The process diffuses at its maximum when $\rho_0$ and $r_0$ are~both high. \figurename{~\ref{fig7b}} shows the trivial effect of $\rho_0$ on the net network energy for high $r_0$, where the triad structural energy is dominant. Around $r_0 \!=\! 0.5$, the network struggles to become balanced as flipping the edge sign conceivably creates more unbalanced triads as compared to other values on the $r_0$ spectrum. \figurename{~\ref{fig8}} shows $E_p(\mathcal{G})$ and $E_{\triangle}(\mathcal{G})$ against $r_0$ for different $(\beta, \kappa)$ values in steady-state. As seen in \figurename{~\ref{fig8a}}, the minimum achievable pairwise and triad energy values increase with $\kappa$ for $r_0 \!<\! 0.6$, whereas the $\rho_{\infty}$ inversely drops. But for larger $r_0$, the impact of $\beta$ and $\kappa$ on the energy and $\rho_{\infty}$ is minimal. That is to say, $E_p(\mathcal{G})$ in \figurename{~\ref{fig8a}} increases to zero as the number of infected users becoming aware rapidly grows with increase in friendly links. Contrarily, $E_\triangle(\mathcal{G})$ in \figurename{~\ref{fig8b}} falls to the~global energy minimum because the triads gradually evolve to become~socially stable as most users have already formed friendly links.\vspace{-0.9em}

\begin{figure}[!t]
\centering
\subfloat[Pairwise energy vs. $r_0$]{\label{fig8a}\includegraphics[width=1.54in]{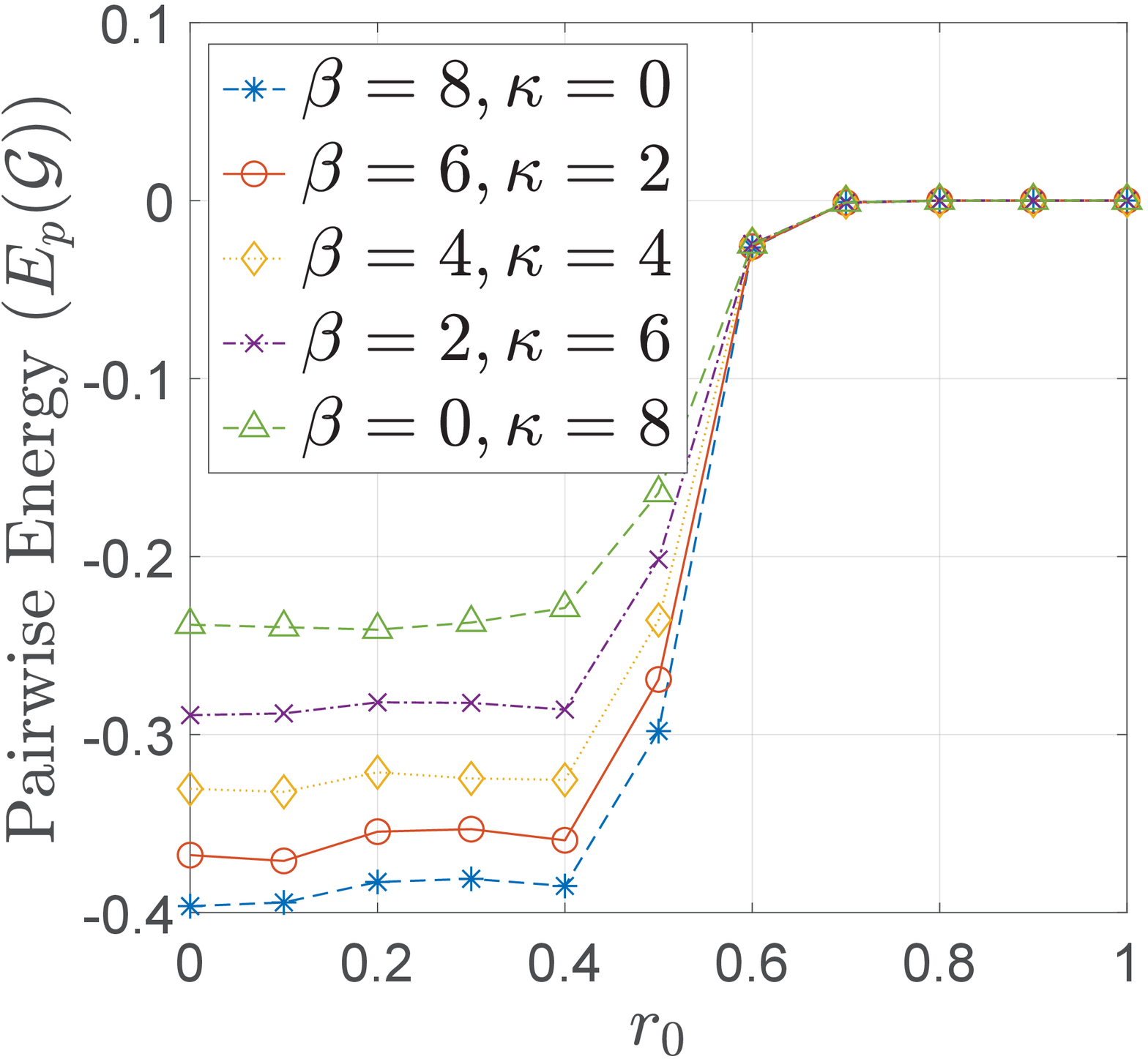}}
~~
\subfloat[Triad energy vs. $r_0$]{\label{fig8b}\includegraphics[width=1.54in]{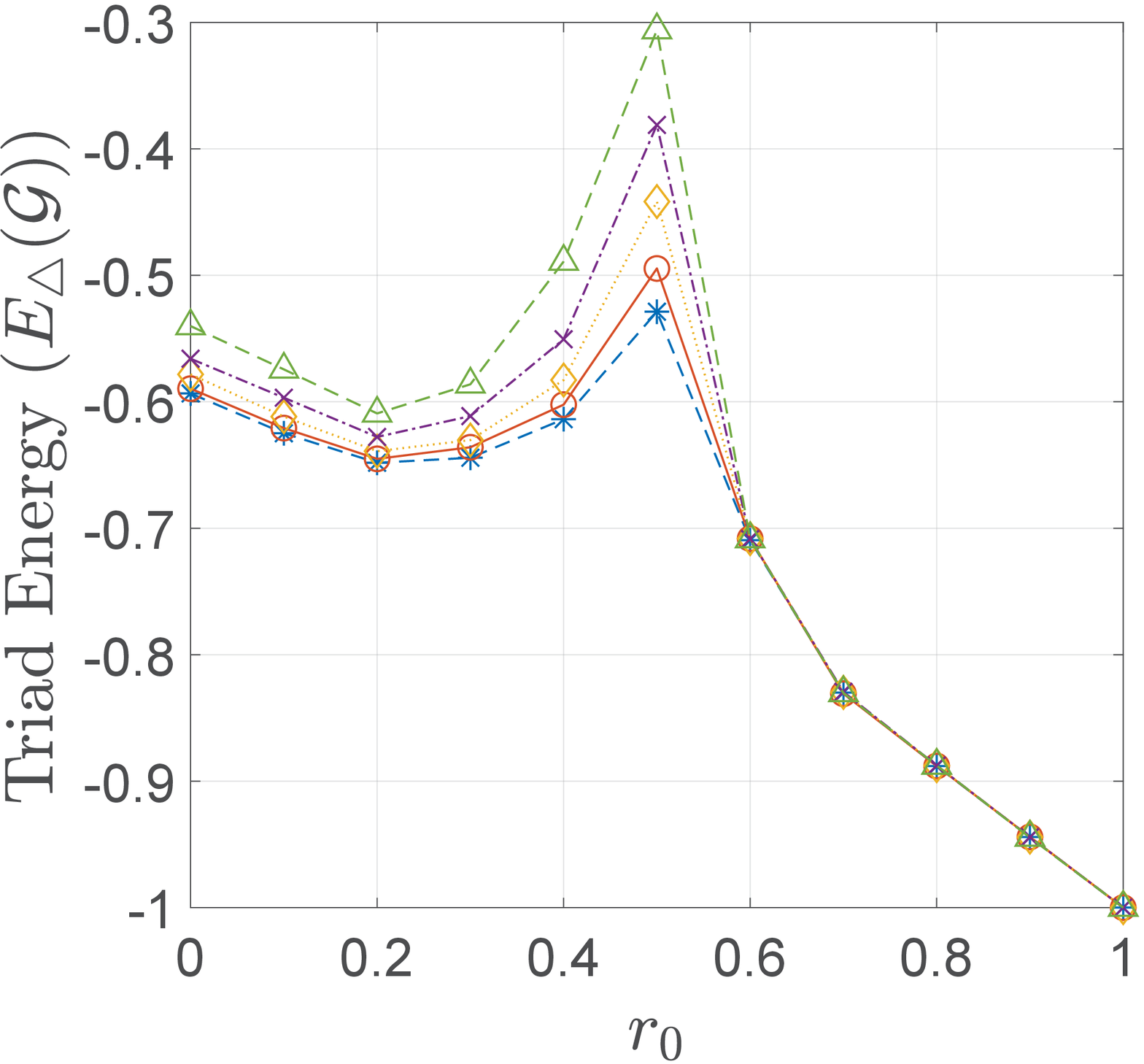}}
\vspace{-0.4em}
\caption{Minimum pairwise and triad energy for different $(\beta, \kappa)$ values. Here, $p_0=0.15$, $\alpha=0.5$, $\beta_a=0.3\beta$, and $\delta=9$.}
\vspace{-0.4em}
\label{fig8}
\end{figure}
\subsection{Case Study II: Results}
\label{sec5.3}
\fontdimen2\font=0.42ex
We now assess our model using the real sparsely-connected SL $(\approx \!1.10\%)$, BC $(\approx \!0.93\%)$, and the dense CS $(\approx\! 74.67\%)$ SN datasets. The total number of triads existing in SL, BC, and CS are $4268$, $1588$, and $74140$,~respectively\footnote{In general, the sharp upper bound on the number of triads in a graph with $n$ nodes and $m$ edges is $\frac{n}{6}(2m-n+1)^{3/2}.$}. Table~\ref{table1} compares the impact of network scalability on the model performance.
%\begin{table}[!t]
%%% increase table row spacing, adjust to taste
%	\renewcommand{\arraystretch}{1.3}
%% if using array.sty, it might be a good idea to tweak the value of
%% \extrarowheight as needed to properly center the text within the cells
%	\caption{Performance under network scalability for $\rho_0 \!=\! 0.15$ and $(\beta,\kappa) \!=\! (4,4)$}
%	\vspace{-1em}
%	\label{table1}
%	\centering
%	\begin{tabular}{|c||c|}
%	\hline
%	One & Two\\
%	\hline
%	Three & Four\\
%	\hline
%	\end{tabular}
%	\vspace{-0.5em}
%\end{table}
\begin{table}[t]
  \renewcommand{\arraystretch}{1.2}
% if using array.sty, it might be a good idea to tweak the value of
% \extrarowheight as needed to properly center the text within the cells
	\caption{Performance under network scalability for $\rho_0 \!=\! 0.15$ and $(\beta,\kappa) \!=\! (4,4)$}
	\vspace{-1.4em}
	\label{table1}
	\centering
  \begin{tabular}{|c|c|c|c|c|c|c|c|}
    \hline
%    \multicolumn{2}{|c|}{} &\multicolumn{5}{|c|}{} \\
%     \cline{3-7}
     \multicolumn{1}{|c|}{}  & $|\mathcal{V}|$ & Density & $\rho_{\infty}$ & $a_{\infty}$ & $r_{\infty}$ & $|\triangle_B|$\\\hline
  \multirow{3}{*}{\rotatebox{90}{SL}}  & $200$ & $\approx 4.37\%$ & $0.22$ & $0.32$ & $0.43$ & $0.987$ \\\cline{2-7}    
   & $498$ & $\approx 1.20\%$ & $0.16$ & $0.16$ & $0.38$ & $0.993$ \\ \cline{2-7}
   & $747$ & $\approx 1.10\%$ & $0.17$ & $0.11$ & $0.37$ & $0.936$ \\
   \hline  \hline
   \multirow{3}{*}{\rotatebox{90}{BC}}  & $161$ & $\approx 4.25\%$ & $0.22$ & $0.39$ & $0.41$ & $0.984$ \\\cline{2-7}    
   & $460$ & $\approx 1.46\%$ & $0.193$ & $0.19$ & $0.39$ & $0.928$ \\ \cline{2-7}
   & $709$ & $\approx 0.93\%$ & $0.17$ & $0.11$ & $0.38$ & $0.906$ \\
   \hline  \hline
   \multirow{1}{*}{\rotatebox{90}{CS}}  & $100$ & $\approx 74.67\%$ & $0.21$ & $0.79$ & $0.66$ & $0.993$ \\
   \hline  
  \end{tabular}
  \vspace{-1em}
\end{table}
\begin{figure*}[t]
\centering
\subfloat[$\rho_{\infty}$ vs. $(\beta,\kappa)$]{\label{fig9a}\includegraphics[width=1.4in]{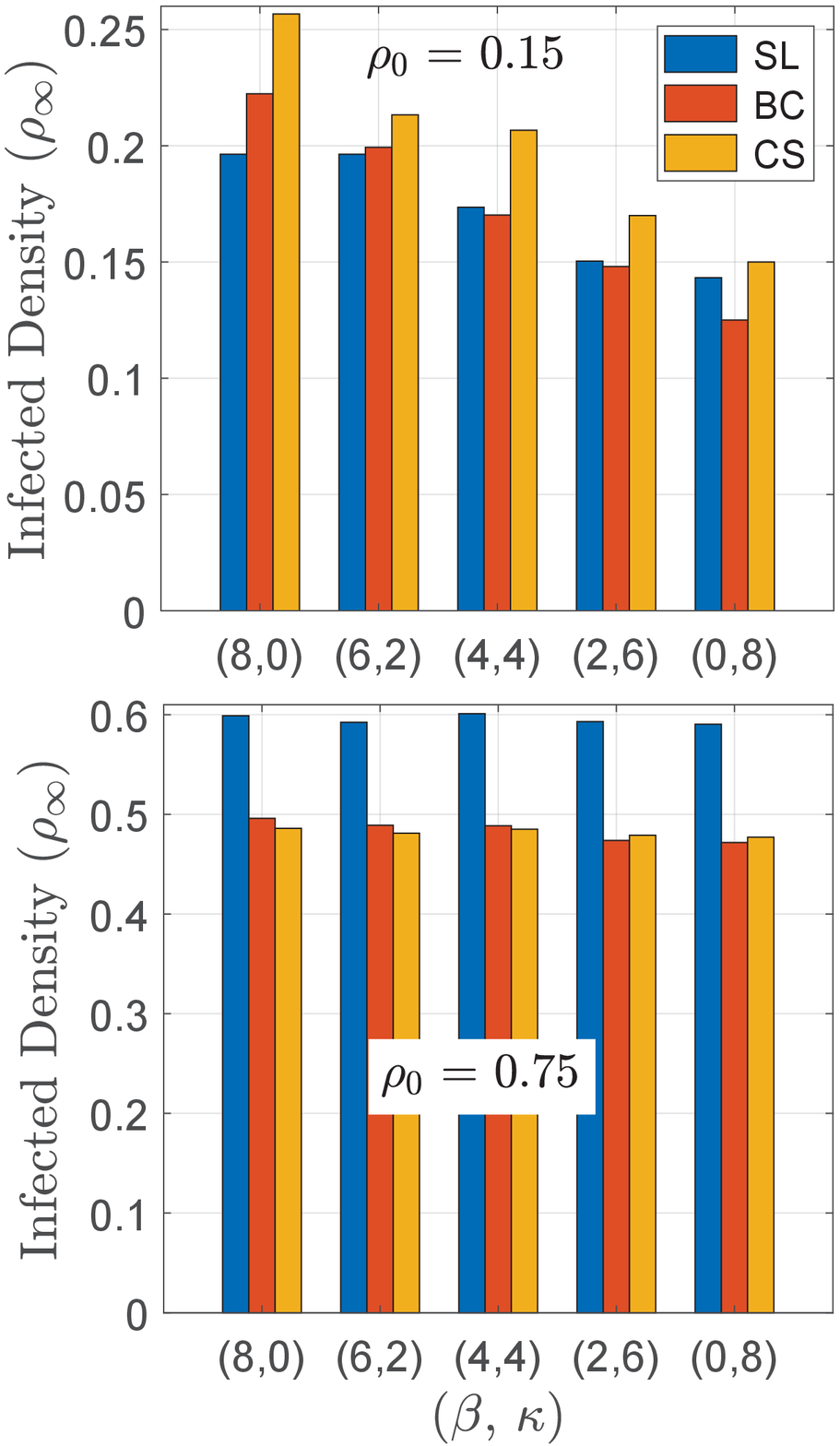}}
\,
\subfloat[$a_{\infty}$ vs. $(\beta,\kappa)$]{\label{fig9b}\includegraphics[width=1.352in]{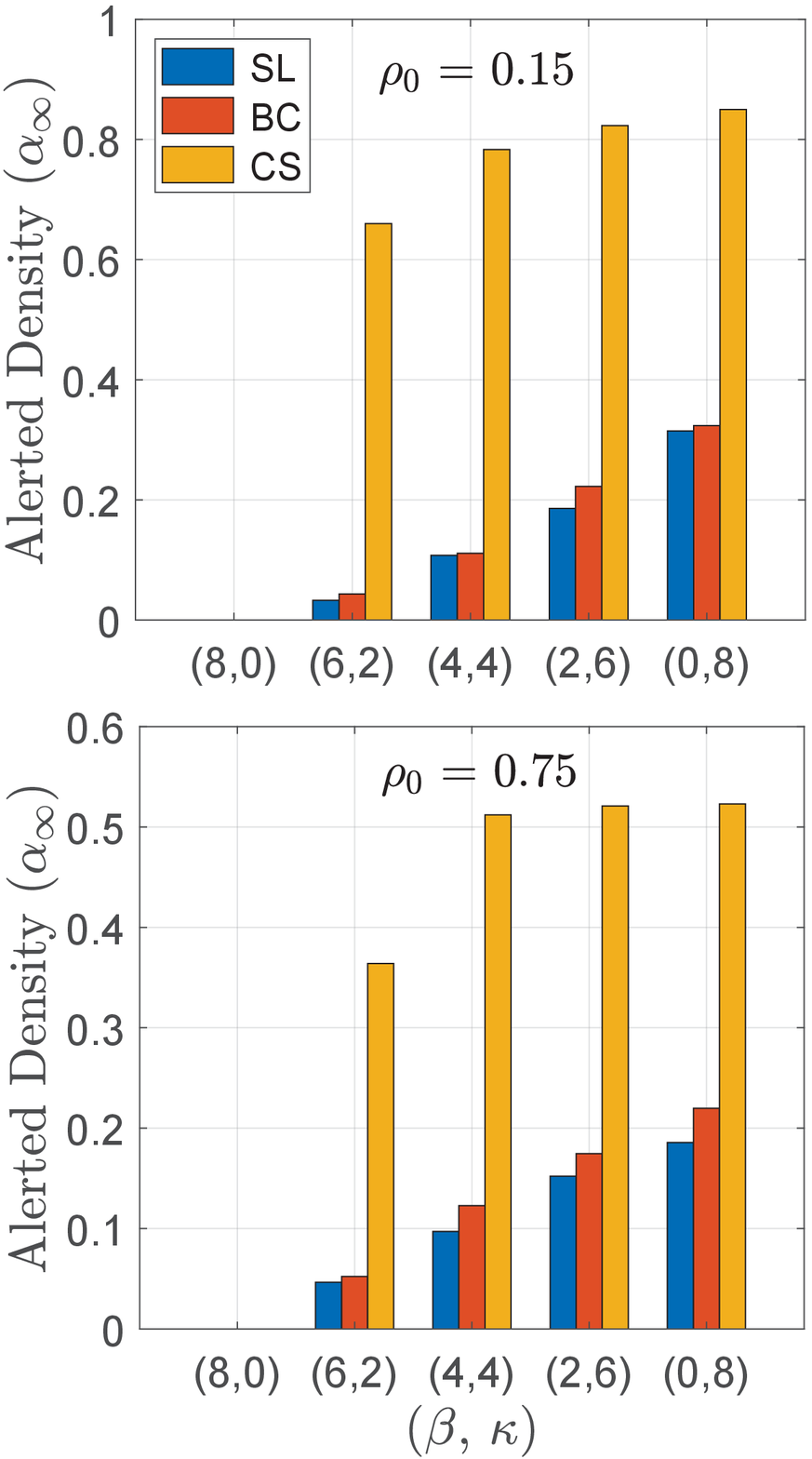}}
\,
\subfloat[$r_{\infty}$ vs. $(\beta,\kappa)$]{\label{fig9c}\includegraphics[width=1.34in]{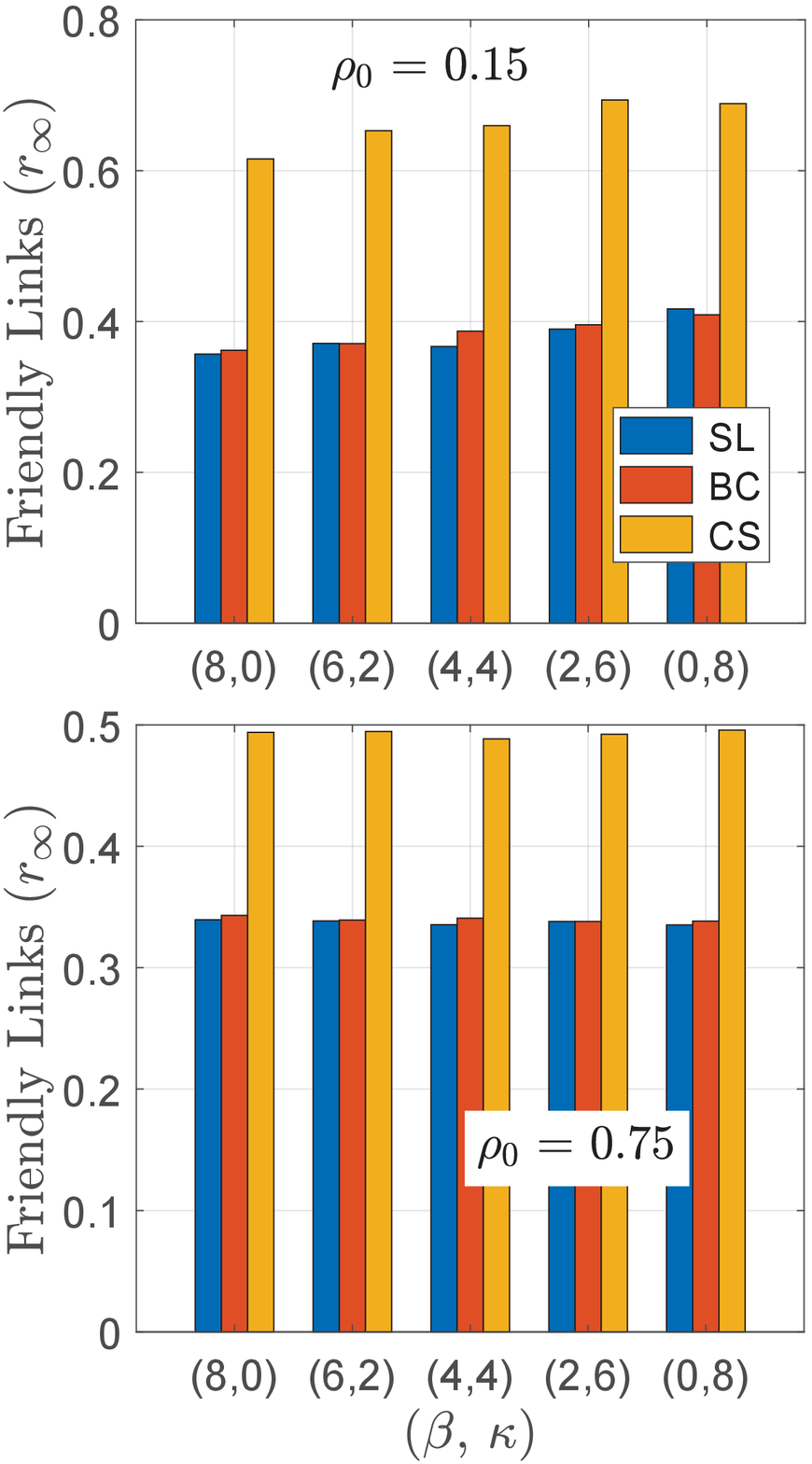}}
\,
\subfloat[$|\triangle_B|$ vs. $(\beta,\kappa)$]{\label{fig9d}\includegraphics[width=1.352in]{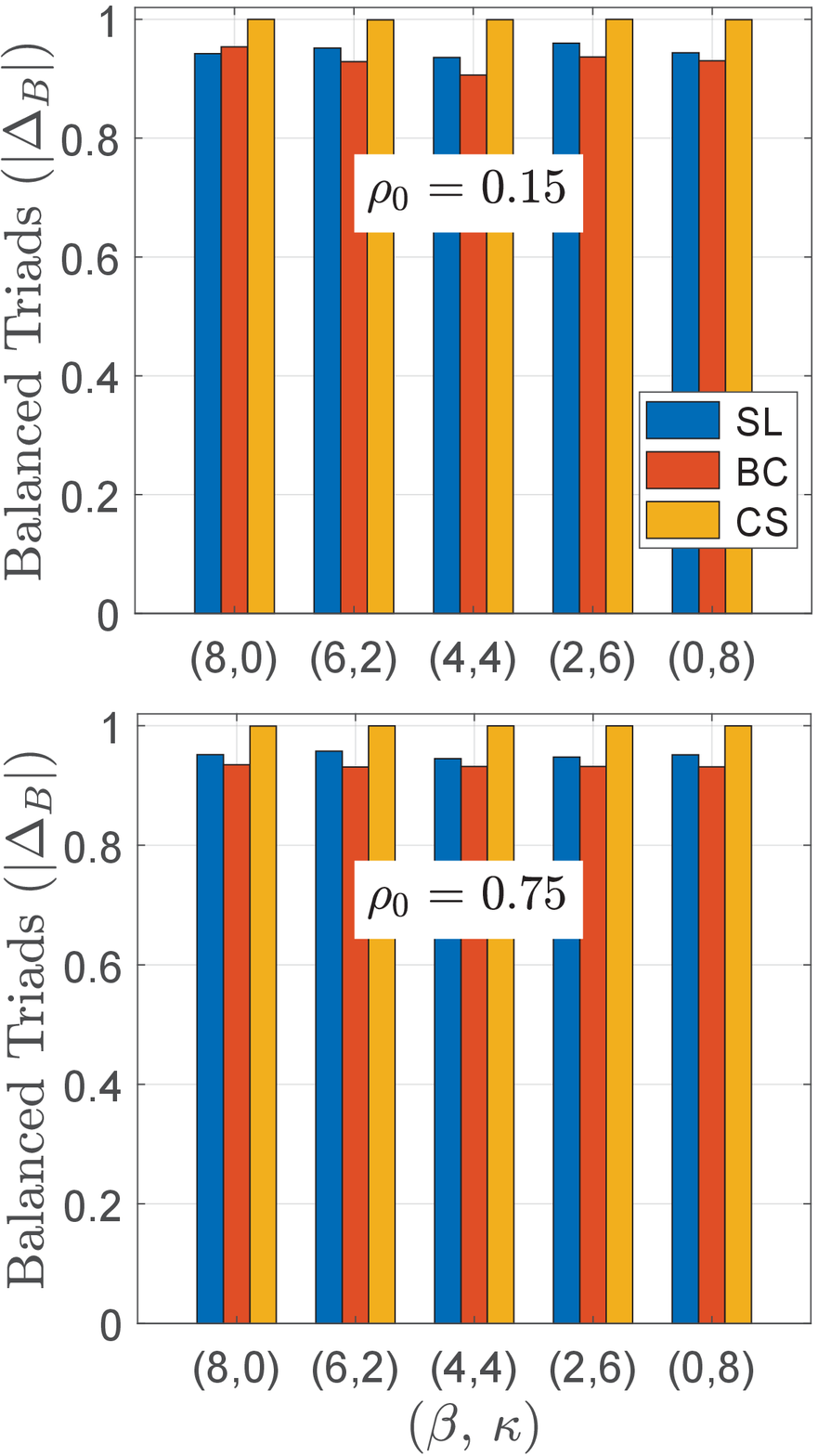}}
\,
\subfloat[$E(\mathcal{G})$ vs. $(\beta,\kappa)$]{\label{fig9e}\includegraphics[width=1.356in]{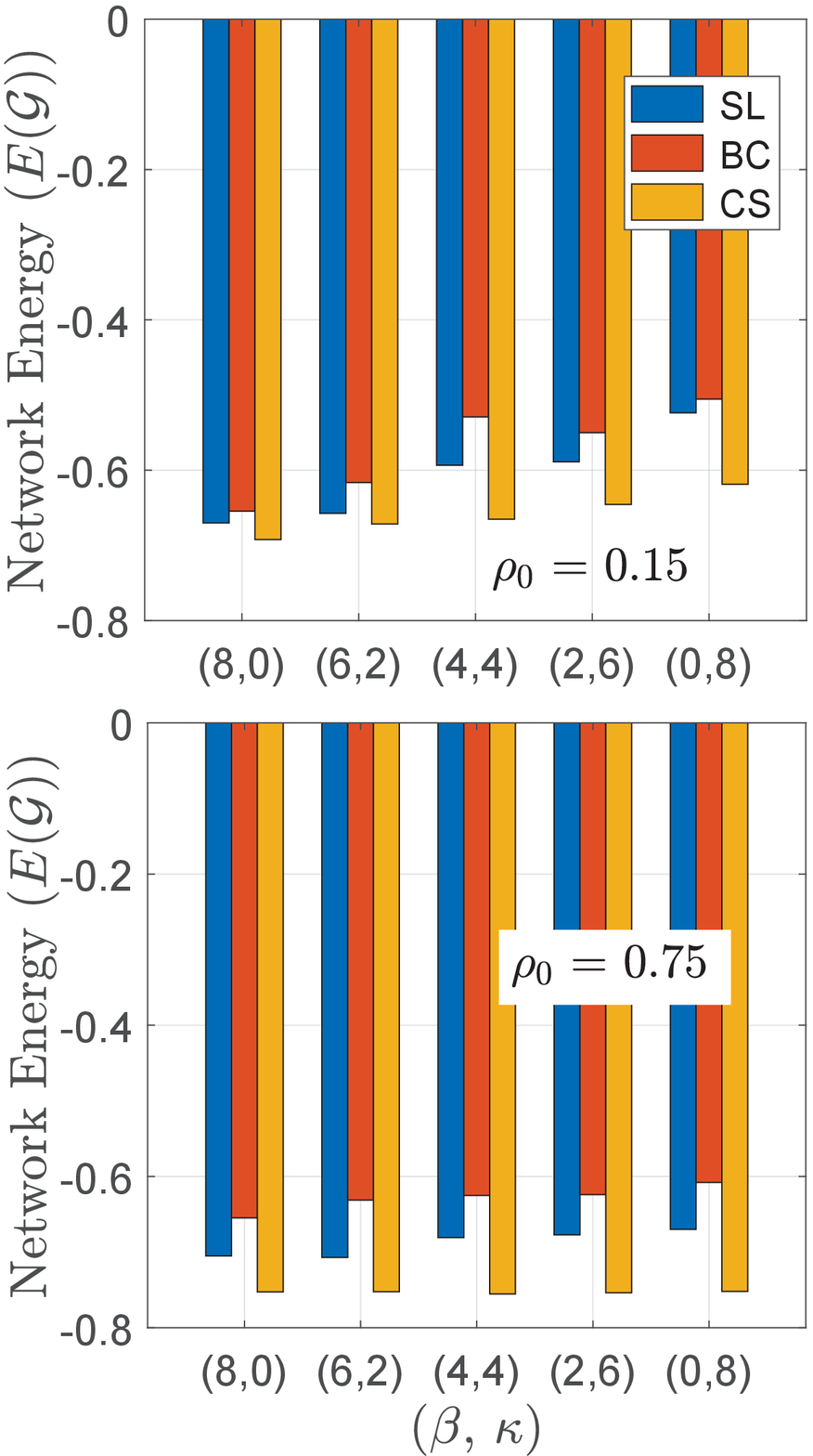}}
\vspace{-0.4em}
\caption{Experimental results on SL, BC, and CS for varying $(\beta,\kappa)$ values. Here, $r_0=0.25$, $p_0=\{0.15,0.75\}$, $\alpha=0.5$, $\beta_a=0.3\beta$, and $\delta=9$.}
\vspace{-0.7em}
\label{fig9}
\end{figure*}

\figurename{~\ref{fig9}} depicts the steady-state results of the datasets with respect to generic $\beta$ and $\kappa$ values. Consistent with \figurename{~\ref{fig6a}}~and \figurename{~\ref{fig7a}}, we see in \figurename{~\ref{fig9a}} that $\rho_{\infty}$ reduces with rise in~awareness for $\rho_0 \!=\!0.15$, while it remains high and almost the~same for $\rho_{0} \!=\! 0.75$. As a resultant, the fraction of alerted users~increases with $\kappa$ as shown in \figurename{~\ref{fig9b}}. It is noteworthy that,~irrespective of $\rho_0$, the gap in $a_{\infty}$ between CS and the other two datasets gradually decreases. Despite the fewer nodes in CS, the~plots reveal the profound impact of $\kappa$ on the nodal states of densely connected CS that contains more number of triads. Summing up the infected and alerted fractions for each ($\beta, \kappa$) pair also justifies the natural immunization induced by the formation of two-cluster networks. Changes in $\beta$ and $\kappa$ however, do~not seem to influence the fraction of friendly links ($r_{\infty}$) in steady-state. Though the epidemic state of nodes are decisive in~link sign evolution, we see that for all scenarios the ultimate~number of friendly ties is almost the same. Upon reaching steady-state, the fraction of balanced triads, denoted by $|\triangle_B|$, is close to $100\%$ in \figurename{~\ref{fig9d}}, which indicates a balanced SN structure in accordance with Heider's balance theory. Furthermore, the results empirically show that although the structural energy of triads approach a global minimum, $E(\mathcal{G})$ does not reach the global energy minimum state. This is because the pairwise energy tends towards a local energy minimum state when $\alpha = 0.5$. Also, since $\kappa$ affects the pairwise spreading energy, $E(\mathcal{G})$ is slightly lower for higher $\kappa$ rates as plotted in \figurename{~\ref{fig9e}}. Comparing the trends in both rows of \figurename{~\ref{fig9}}, the results for $r_{\infty}$ and $|\triangle_B|$ are almost alike. With increase in $\kappa$, $E(\mathcal{G})$ of CS is higher by roughly $28\%$ for $\rho_0=0.15$ because of the larger population of alerted users being intrinsically immunized as compared to the setting with $\rho_0 = 0.75$, where the impact of increasing $\kappa$ is relatively inconsequential.\vspace{-0.8em}

\section{Conclusion}
\label{sec6}
\fontdimen2\font=0.42ex
In this paper, coupled dynamics of the SAIS epidemic model and the structural evolution of SSNs was studied. Inspired by Heider's balance theory, a network energy framework~was formulated to capture the viral spreading via pairwise user~interactions in conjunction with social stability in triad configurations. The superiority of incorporating user awareness in the classical SIS model was fully validated~by the Monte Carlo simulation results. Moreover, it was shown that a complete SSN splits into two clusters of alerted and infected users~upon reaching a local energy minimum. The alerted cluster density was also found to grow with increase in the initial number of friendly links and a fully balanced SSN becomes infection-free only when the triad energy is considered and all initial user links are friendly. One interesting future work is to leverage network centrality measures other than degree distribution in probing the trade-off between opposing epidemics and social stability in directed and composite SSNs. Better~parameter estimations based on network spectral analysis is also of value in adopting effective control strategies.\vspace{-0.8em}
\ifCLASSOPTIONcompsoc
  \section*{Acknowledgments}
\else
  \section*{Acknowledgment}
\fi

\fontdimen2\font=0.5ex
This research was supported by the Faculty Development Competitive Research Grant (No.\,240919FD3918), NU.\vspace{-0.4em}

\ifCLASSOPTIONcaptionsoff
  \newpage
\fi



\bibliographystyle{IEEEtran}
 %argument is your BibTeX string definitions and bibliography database(s)
\bibliography{IEEEabrv,reference}
\end{document}